\documentclass[12pt,a4paper]{article}
\usepackage[utf8x]{inputenc}
\usepackage[T1]{fontenc}
\usepackage{amsfonts}
\usepackage{amsmath}
\usepackage[pdftex]{graphicx}
\usepackage[pdftex,linkcolor=black,pdfborder={0 0 0}]{hyperref} 
\usepackage{calc} 
\usepackage{enumitem} 
\usepackage{amsthm}
\usepackage{bm}
\usepackage{xcolor}
\usepackage{natbib}
\usepackage{placeins}
\usepackage{comment}
\usepackage{algorithm2e}
\RestyleAlgo{ruled}
\usepackage{dsfont}
\usepackage[makeroom]{cancel}
\usepackage{pdfpages}
\usepackage{multirow}
\usepackage{adjustbox}
\usepackage{amssymb}
\usepackage{cancel}
\usepackage{caption}
\usepackage[compact]{titlesec}
\usepackage{etoolbox} 
\usepackage{stmaryrd}

\newtheorem{theorem}{Theorem}

\def\spacingset#1{\renewcommand{\baselinestretch}%
{#1}\small\normalsize} \spacingset{1}

\def \smi {\mathrm{SMI}}

\newtheorem{assumption}{Assumption}
\newtheorem{lemma}{Lemma}

\DeclareMathOperator{\HND}{HN}

\DeclareMathOperator{\ND}{\mathcal N}
\DeclareMathOperator{\skewND}{\mathcal{SN}}

\DeclareMathOperator{\UD}{U}

\newcommand\indicator[1]{\mathds{1}\left\{ #1 \right\}}

\newcommand{\E}{\mathbb{E}}
\newcommand{\KL}{\mathrm{KL}}
\newcommand{\dt}{\mathrm{d}}
\newcommand{\MD}{\mathsf{L}}
\newcommand{\argmin}{\operatornamewithlimits{argmin}} 
 
\DeclareMathOperator*{\bigtimes}{\vartimes}

\begin{document} 

\begin{titlepage}

\title{Bayesian Modular Inference for Copula Models with Potentially Misspecified Marginals}

\author{$\mbox{Lucas Kock}^1$, $\mbox{David T. Frazier}^2$, $\mbox{Michael Stanley Smith}^{3,\ast}$, and $\mbox{David J. Nott}^1$} 

\date{\today}
\maketitle
\thispagestyle{empty}
\noindent

\begin{center}
{\Large Abstract}
\end{center}
\vspace{-1pt}
\noindent 
Copula models of multivariate data are popular because they allow 
separate specification of marginal distributions and the copula function.
These components can be treated as inter-related modules in a modified Bayesian inference approach called ``cutting feedback''
that is robust 
to their misspecification. 
Recent work uses a two module approach, where all $d$
marginals form a single module, to robustify inference for the marginals 
against copula function misspecification, or vice versa. 
However, marginals can exhibit differing levels of misspecification, 
making it attractive to assign each its own module with an individual 
influence parameter controlling its contribution to a joint 
semi-modular inference (SMI) posterior. This 
generalizes existing two module SMI methods, which interpolate between 
cut and conventional posteriors using a single influence parameter.
We develop a novel copula SMI method and select the influence parameters using 
Bayesian optimization. It provides an efficient continuous relaxation of 
the discrete optimization problem over $2^d$ cut/uncut configurations. 
We establish theoretical properties of the resulting semi-modular posterior 
and demonstrate the approach on simulated and real data. The real data 
application uses a skew-normal copula model of asymmetric dependence 
between equity volatility and bond yields, where robustifying copula 
estimation against marginal misspecification is strongly motivated.

\vspace{20pt}
 

\vspace*{\fill}
\noindent {\small\textbf{Acknowledgments:} David Nott's research was supported by the Ministry of Education, Singapore, under the Academic Research Fund Tier 2 (MOE-T2EP20123-0009), and he is affiliated with the Institute of Operations Research and Analytics at the National University of Singapore. Michael S. Smith and David T. Frazier gratefully acknowledge support by the Australian Research Council through a Discovery Project grant DP250101069.}

\vspace{20pt}

\noindent{\small
$^1$ Department of Statistics and Data Science, National University of Singapore\\
$^2$ Department of Econometrics and Business Statistics, Monash University\\
$^3$ Melbourne Business School, University of Melbourne\\
$^\ast$ Correspondence should be directed to Michael Smith at {\tt mike.smith@mbs.edu}\,.  
}

\end{titlepage}

\spacingset{2.0} 
\setlength{\abovedisplayskip}{3pt}
\setlength{\belowdisplayskip}{3pt}
\setlength{\abovedisplayshortskip}{0pt}
\setlength{\belowdisplayshortskip}{0pt}

\section{Introduction}
Copula models of multivariate continuous data are very popular because they
allow the marginal distributions and copula function characterizing the dependence
structure to be specified separately. However, in many problems specification of the copula function
or the 
marginals can be difficult. Here we consider modified 
Bayesian inference for copula models, in settings where we wish to 
robustify inference about the copula parameters
to misspecification of marginal distributions. 
Our work contributes to the literature on modular Bayesian inference, 
which decomposes a joint model into simpler inter-related sub-models called modules, and
attempts to limit the influence that misspecified modules can have on parameters
in well-specified ones. We consider each of the $d$ marginals as a separate module, and 
develop a ``semi-modular inference''~\citep[SMI;~][]{CarNic2020,CarNic2022,FraNot2025} method that robustifies inference
for the copula parameters from potential misspecification of one or more of the marginals.
We establish the theoretical properties of the resulting semi-modular posterior (SMP) and demonstrate the 
approach empirically. 

A common modular Bayesian inference approach is
to calculate a ``cut posterior'', which is where feedback from misspecified modules
to the parameters of other well-specified modules is completely cut~\citep{Plu2015,liu+g25}. In recent work, \cite{SmiYuNotFra2025}
develop cut posteriors for copula models, but treat all marginals as a single module. 
However, in practice marginals can exhibit differing levels of misspecification, 
making it attractive to treat them as $d$ separate modules, and
to modulate the influence of each according to the degree of
misspecification.  

SMI \citep{CarNic2020,CarNic2022,FraNot2025}
extends cutting feedback approaches to the case
where the effects of a misspecified module on a joint SMP are modulated continuously 
by an influence parameter, which
interpolates between a cut and the conventional posterior. Here, we specify a completely novel SMI approach for copulas, with a 
separate influence parameter for each marginal distribution. These influence parameters are introduced through a novel extended pseudo likelihood, which is a significant departure from the SMI approach of~\cite{CarNic2020,CarNic2022}. We develop a Bayesian
optimization \citep[BO;~][]{garnett23} approach to learn the influence
parameters of this new SMP, which allows some marginals
to be fully cut, some fully uncut, and some partially cut.  Even when searching
for a modular structure where each marginal is fully cut or uncut, the SMI
optimization is an effective continuous relaxation of the discrete search problem that is intractable even in modest dimensions.

We study the
properties of our SMP both theoretically and empirically. Re-interpreting the SMP as a type of generalized Bayesian posterior (\citealp{bissiri2016general}), we can view the influence parameters in the SMI framework as akin to the ``learning rates'' required to perform generalized Bayesian inference. However, in contrast to standard generalized Bayes, our theoretical results demonstrate that our SMP depends critically on the value of the influence parameters, which determine the extent of posterior modularity. In general, the learning rates in generalized Bayes posteriors only impacts 
the posterior scale, while the influence parameters in the SMP impact both posterior scale and location. This makes existing results for generalized posteriors inapplicable, so that we  must tailor our theoretical analysis to this novel setting. 
The strong dependence on the influence parameters also suggests that ---in contrast to the case of learning rates in generalized Bayes--- values for the influence parameters can be learned that deliver accurate results, where accuracy is 
measured by an external criterion.
A simulation study illustrates the impact of the influence parameters empirically. Considering a bivariate example, where only one of the marginals is misspecified, we show that tuning the influence parameter to down-weight the impact of the misspecified module improves inference for the copula and the misspecified marginal at the cost of a mild deterioration for the well-specified marginal. In this example, an SMP with a partial cut of the misspecified module is preferred over both fully cut and uncut posteriors. An additional high-dimensional example is given in the Supporting Information.

In a substantive application to U.S. financial data, we estimate a copula model of equity market volatility and the yields on domestic debt graded AAA and BBB by Standard \& Poors. A 6-dimensional skew-normal copula captures strongly asymmetric contemporaneous and serial
dependence between the variables. Parametric skewed distributions with heavy tails are used for the marginals, which are preferred over nonparametric ones because the extreme tail behavior is important. However, one or more of these marginals may be misspecified, so that computing inference robust to marginal misspecification is attractive. We show that the SMP differs substantially
from the conventional and fully cut posteriors, with results that are consistent with both the data and economic intuition. 

The rest of the paper is organized as follows.
Section~\ref{sec:background} gives the necessary background on copula models, variational inference and Bayesian modular inference methods. Section~\ref{sec:genpost} briefly outlines the
cutting feedback approach for copula models of~\cite{SmiYuNotFra2025}, how to generalize it to allow each
marginal to be a separate module, and outlines the proposed SMI approach.    
Efficient variational inference methods to compute the SMP are then given, where Bayesian optimization 
is used to learn the influence parameters.
Section~\ref{sec:theory} explores the properties of our SMP, 
Section~\ref{sec:simulation} presents the simulation study, Section~\ref{sec:yieldseg}~contains the financial application, and Section~\ref{sec:conc} concludes. 

\section{Background}\label{sec:background}
\subsection{Copula models}
%Copulas \citep[e.g.;~][]{Nel2006} give a convenient way to decompose multivariate distributions where the marginal distributions and dependence structure can be specified separately. 
Let $Y=(Y_1,\dots,Y_d)^\top$ be a random vector with joint distribution $F_Y$ and marginal distributions $F_j$, $j=1,\dots,d$. Sklar's theorem \citep{Skl1959} shows that the joint distribution function is
\begin{equation}\label{eq:copula}
    F_Y(y)=C(F_1(y_1),\dots,F_d(y_d)),
\end{equation}
where $C:[0,1]^d\to[0,1]$ is a copula function determining the dependence structure of $Y$. The copula function is unique if all marginals are continuous, and differentiating \eqref{eq:copula} gives the joint density of $Y$,
\begin{align}\label{eq:joint_pdf}
    f_Y(y)=c(F_1(y_1),\dots,F_d(y_d))\prod_{j=1}^d f_j(y_j),
\end{align}
where $c(u)=\frac{\partial}{\partial u}C(u)$ is the ``copula density'', and $f_j(y_j)=\frac{\partial}{\partial y_j}F_j(y_j)$, $j=1,\dots,d$, are the marginal densities. 

A major advantage of copula models is that the marginals and copula function can be specified separately. In many applications parametric families are employed for both, and we write $C(u;\psi)$ for the copula function with parameters $\psi$, and $F_j(y_j;\eta_j)$ for the $j$-th marginal with parameters $\eta_j$. The model parameters are thus $\theta=(\psi^\top,\eta^\top)^\top$, where $\eta=(\eta_1^\top,\dots,\eta_d^\top)^\top$ are the marginal parameters. We write $\mathcal{D}=\{y_1,\dots,y_n\}$ for a set of $i=1,\ldots,n$ independent observations $y_i=(y_{i1},\dots,y_{id})^\top \sim F_Y$. Under the prior $p(\theta)=p(\psi)p(\eta_1)\cdots p(\eta_d)$, the joint posterior is then 
\begin{align*}
   p(\theta\mid\mathcal{D})\propto p(\theta)\prod_{i=1}^n c(F_1(y_{i1};\eta_1),\dots,F_d(y_{id};\eta_d);\psi)\prod_{j=1}^d f_j(y_{ij};\eta_j).
\end{align*}

\subsection{Variational inference}
Variational inference \citep[VI;~][]{BleKucMca2017} is a computationally efficient alternative to exact Bayesian inference computed via Markov chain Monte Carlo (MCMC). The main idea is to approximate the posterior  $p(\theta|\mathcal{D}) \propto p(y|\theta)p(\theta)$ by a distribution $q_\lambda(\theta)$ chosen from a tractable variational family $\mathcal{Q}=\{q_\lambda(\theta)\mid\lambda\in\Lambda\}$ parameterized by $\lambda$.
An optimal value $\lambda^*\in\Lambda$ is commonly obtained by minimizing the reverse Kullback-Leibler (KL) divergence, 
\begin{equation}\label{eq:kl}
\KL(q_\lambda(\theta)\|p(\theta\mid \mathcal{D}))=\mathbb{E}_{q_\lambda(\theta)}\left[\log q_\lambda(\theta)-\log p(\theta\mid \mathcal{D})\right],
\end{equation}
where $\mathbb{E}_{q_\lambda(\theta)}$ denotes expectation with
respect to $q_\lambda(\theta)$.  Minimizing \eqref{eq:kl} is equivalent to maximizing the evidence lower bound \citep[ELBO;~][]{OrmWan2010}
\begin{equation*}
    \mathcal{L}(\lambda)=\mathbb{E}_{q_\lambda(\theta)}\left[\log (p(\theta)p(y\vert\theta))-\log q_\lambda(\theta)\right]\,.
\end{equation*}
A general approach to optimizing the ELBO is to employ stochastic gradient methods. These update an initial value $\lambda^{(0)}$ using an iterative scheme of the form $\lambda^{(m+1)}=\lambda^{(m)}+\rho^{(m)}\circ\widehat{\nabla_\lambda\mathcal{L}\left(\lambda^{(m)}\right)}$ at each step $m$, where $\rho^{(m)}$ is an adaptive vector-valued step size, $\circ$ denotes element-wise multiplication, and $\widehat{\nabla_\lambda\mathcal{L}\left(\lambda^{(m)}\right)}$ is an unbiased estimate of the gradient of $\mathcal{L}(\lambda)$ evaluated at $\lambda^{(m)}$. 

While $\nabla_\lambda\mathcal{L}\left(\lambda^{(m)}\right)$ can be estimated by Monte Carlo estimation of the expectation, variance reduction techniques such as the reparameterization trick \citep{KinWel2014,RezMohWie2014} are often employed to ensure fast convergence. The reparameterization trick assumes that samples $\theta\sim q_\lambda(\theta)$ can be generated as $\theta=f(\varepsilon,\lambda)$ for a deterministic differentiable function $f$ and a random variable $\varepsilon\sim\pi(\varepsilon)$ from a base distribution not depending on $\lambda$. Then,
\begin{align*}
\nabla_\lambda \mathcal{L}(\lambda)
  %  \nabla_\lambda \mathbb{E}_{q_\lambda(\theta)}\left[\log (p(\theta)p(y\vert\theta))-\log q_\lambda(\theta)\right] 
  = \mathbb{E}_{\pi(\varepsilon)}\left[\nabla_{\lambda} \left\{\log (p(f(\varepsilon,\lambda))p(y\vert f(\varepsilon,\lambda)))-\log q_\lambda(f(\varepsilon,\lambda))\right\}\right].
\end{align*}
A single sample from $\pi(\varepsilon)$ in each iteration $m$ is often enough to approximate the ELBO gradient. The choice of step sizes $\rho^{(m)}$ is usually made adaptively, with popular approaches including ADAM \citep{KinBa2014} and ADADELTA \citep{Zei2012}. VI has been used to estimate copula models by~\cite{LoaizaMayaSmith2019} and~\cite{NguyenAusinGaleano2020}.

\subsection{Modular inference under misspecification}\label{sec:modinf}
In this paper we consider Bayesian modular inference methods for
misspecified copula models. Modular inference is used for statistical models which are formed by interacting sub-models that are called ``modules''. Cutting feedback \citep[e.g.~][]{Plu2015,liu+g25} is a modular Bayesian inference method that aims to cut feedback from misspecified modules to well specified ones. It has been considered previously by \citet{SmiYuNotFra2025} for copula models. The current work extends cutting feedback for copulas to semi-modular inference (SMI), and we first briefly outline both cutting feedback and SMI. We discuss these for two modules for clarity, but our later semi-modular extensions treat each marginal distribution in a copula model as its own module, and hence involve many modules. A general approach to cutting feedback with multiple modules is given in \citet{liu+g25}.

Consider a model for data $\mathcal{D}$ with density $p(\mathcal{D}\mid\theta)$ and parameter vector $\theta=(\psi^\top,\eta^\top)^\top$ that can be factorized as 
\begin{equation}\label{eq:two_module_system}
    p(\mathcal{D}\mid\theta) = p_1(\mathcal{D}\mid \psi)p_2(\mathcal{D}\mid \eta,\psi)\,.
\end{equation}
In many cutting feedback applications, the terms $p_1$ and $p_2$ are likelihoods for different data sources for the two modules, but in general they only need to represent different factors in a decomposition of the joint likelihood. We show later that this is the case with copula models. Module~1 comprises $p_1(\mathcal{D}\mid \psi)$ and prior $p(\psi)$, and module~2 of $p_2(\mathcal{D}\mid \eta,\psi)$ and prior $p(\eta\mid\psi)$. 
The main idea of cutting feedback is that module~1 is trusted, but module~2 is not, and we wish to avoid the misspecification in module 2 corrupting inferences about the parameters in module 1.

The joint posterior is $p(\psi,\eta\mid\mathcal{D})=p(\psi\mid\mathcal{D})p(\eta\mid\psi,\mathcal{D})$, where 
%\begin{align*}
%    p(\psi\mid\mathcal{D})&\propto p(\psi)p_1(\mathcal{D}\mid \psi)\Bar{p}_2(\mathcal{D}\mid \psi)\,,\\
%    p(\eta\mid\psi,\mathcal{D})&\propto p(\eta\mid\psi)p_2(\mathcal{D}\mid \eta,\psi)\,.
%\end{align*}
\begin{equation*}
    p(\psi\mid\mathcal{D})\propto p(\psi)p_1(\mathcal{D}\mid \psi)\Bar{p}_2(\mathcal{D}\mid \psi)\;,\mbox{ and }\;
    p(\eta\mid\psi,\mathcal{D})\propto p(\eta\mid\psi)p_2(\mathcal{D}\mid \eta,\psi)\,.
\end{equation*}
The term $\Bar{p}_2(\mathcal{D}\mid \psi)=\int p(\eta\mid\psi)p_2(\mathcal{D}\mid \eta,\psi)d\eta$ is called the feedback term because it represents the feedback of module~2 on the marginal for $\psi$. Cutting feedback is a way to stop misspecification in the second module from corrupting inference for parameters in the trusted module 1. The feedback term is deleted from $p(\eta\mid\psi,\mathcal{D})$ to define the cut posterior 
\begin{align*}
    p_\text{cut}(\psi,\eta\mid\mathcal{D})= p_\text{cut}(\psi\mid\mathcal{D})p(\eta\mid\psi,\mathcal{D})\propto p_\text{cut}(\psi\mid\mathcal{D})p(\eta\mid\psi)p_2(\mathcal{D}\mid \eta,\psi),%\,,\mbox{where}\\ p_\text{cut}(\psi\mid\mathcal{D})&\propto p(\psi)p_1(\mathcal{D}\mid \psi)\,.
\end{align*}
where $p_\text{cut}(\psi\mid\mathcal{D})\propto p(\psi)p_1(\mathcal{D}\mid \psi)$.
\citet{CarNic2020,CarNic2022} extend this idea to SMI, which introduces an influence parameter $\gamma$ to smoothly interpolate between the conventional and cut posteriors. This approach offers several advantages. First, even under model misspecification, allowing some information flow from the misspecified to the trusted module can be beneficial; see~\citet{FraNot2025} for the bias-variance trade-off involved. Second, in a model with $m>2$ modules, deciding which to cut can be challenging as it involves a search over $2^m$ possible configurations. 
Later we show that relaxing this choice into a continuous search problem over a multivariate influence parameter $\gamma$ simplifies it greatly. 

\citet{CarNic2022} define a semi-modular posterior (SMP) using the power posterior
\begin{equation}\label{eq:powerpostCN}
p_{\text{pow},\gamma}(\psi,\Tilde{\eta}\mid\mathcal{D})\propto p(\psi)p_1(\mathcal{D}\mid \psi)p_2(\mathcal{D}\mid \Tilde{\eta},\psi)^\gamma p(\Tilde{\eta}\mid\psi)\,,
\end{equation}
which introduces auxiliary parameters $\tilde \eta$ to define the joint 
\begin{equation*}
    p_{\text{SMI},\gamma}(\psi,\Tilde{\eta},\eta\mid\mathcal{D})\propto p_{\text{pow},\gamma}(\psi,\Tilde{\eta}\mid\mathcal{D})p(\eta\mid\psi)p_2(\mathcal{D}\mid \eta,\psi)\,.
\end{equation*}
Integrating out $\tilde \eta$ gives the SMP 
$p_{\text{SMI},\gamma}(\psi,\eta\mid\mathcal{D})=\int p_{\text{SMI},\gamma}(\psi,\Tilde{\eta},\eta\mid\mathcal{D})d\Tilde{\eta}$.
%\begin{align}
%    p_{\text{SMI},\gamma}(\psi,\Tilde{\eta},\eta\mid\mathcal{D})&=p_{\text{pow},\gamma}(\psi,\Tilde{\eta}\mid\mathcal{D})p(\eta\mid\psi)p_2(\mathcal{D}\mid \eta,\psi) \nonumber \\ 
%    p_{\text{pow},\gamma}(\psi,\Tilde{\eta}\mid\mathcal{D})&\propto p(\psi)p_1(\mathcal{D}\mid \psi)(p_2(\mathcal{D}\mid \psi,\Tilde{\eta}))^\gamma p(\Tilde{\eta}\mid\psi). \label{eq:powpost_carnic}
%\end{align}
Setting $\gamma=1$ results in the conventional posterior,  $\gamma=0$ gives the cut posterior $p_\text{cut}(\psi,\eta\mid\mathcal{D})$, whereas $0<\gamma<1$ interpolates between the two. In this paper, we later propose a variation of SMPs for copula models that is not based on the power posterior.

The original motivation of \cite{CarNic2020} for introducing SMI was that the 
conventional posterior can be preferable to the cut posterior when 
misspecification is only moderate. Their crucial insight was to frame the 
problem in terms of a bias-variance trade-off, with the influence parameter 
managing this trade-off by interpolating between the cut and conventional 
posteriors. \cite{FraNot2025} formalized this intuition for a variant 
of SMI introduced in \cite{ChaNotDroFraSis2023}, and \cite{Nicholls+lwc22} 
develop connections between SMPs and generalized Bayesian 
approaches; related generalized Bayesian treatments of cutting feedback 
appear in \cite{frazier2025cutting} and \cite{Tan+nf25}. Selecting 
influence parameters has been studied in \cite{CarNic2022} and
\cite{battaglia+chlbn25} using amortized variational inference, and we 
also adopt variational methods here to compute SMPs, 
following \cite{YuNotSmi2023}.  

\section{Generalized posteriors for copula models with misspecified marginals}\label{sec:genpost}

\subsection{Cutting feedback in copula models}\label{sec:SmiYuNotFra2025}
\citet{SmiYuNotFra2025} consider copula models as a two module system, where the copula function $C(u;\psi)$ forms one module, and the marginals $F_1(y_1;\eta_1)$, \dots, $F_d(y_d;\eta_d)$ collectively form a second module. They consider a scenario where a researcher is confident in the copula function, but the marginals may be misspecified, and therefore wishes to cut the influence of $\eta$ on $\psi$. They call this a ``Type 2'' cut posterior and define it using a pseudo likelihood of the rank data to bring the likelihood into the form \eqref{eq:two_module_system}. Denote the rank data as $r(\mathcal{D})=\{r(y_{ij}); i=1,\dots,n, j =1,\dots,d\}$, where $r(y_{ij})=\sum_{k=1}^n\indicator{y_{kj}\leq y_{ij}}$ is the rank of $y_{ij}$ within marginal $j$.  Setting 
\begin{align*}
    a_{ij}(\mathcal{D}) = \frac{r(y_{ij})-1}{n+1}\;\;\mbox{ and }\;\;
    b_{ij}(\mathcal{D}) = \frac{r(y_{ij})}{n+1}\,,
\end{align*}
\citet{SmiYuNotFra2025} define the pseudo rank likelihood as
\begin{equation}\label{eq:pl}
    p_\text{pl}(r(\mathcal{D})\mid\psi)= \prod_{i=1}^{n}\Delta_{a_{i1}(\mathcal{D})}^{b_{i1}(\mathcal{D})}\dots \Delta_{a_{id}(\mathcal{D})}^{b_{id}(\mathcal{D})}C(u_i;\psi),
\end{equation}
where for $v=(v_1,v_2,\ldots,v_d)^\top$, 
\begin{equation*}
    \Delta_{a_{j}}^{b_{j}}C(v;\psi)=C(v_{1},\dots,v_{j-1},b_{j},v_{j+1},\dots,v_d)-C(v_{1},\dots,v_{j-1},a_{j},v_{j+1},\dots,v_{d})
\end{equation*}
is a differencing operator over the $j$th element. Equation~\eqref{eq:pl} is
similar to replacing each marginal in the copula model with its empirical distribution function to remove dependence on the marginals, but accounting 
for the discrete nature of the ranks.  By setting 
$p_1(\mathcal{D}|\psi)=p_\text{pl}(r(\mathcal{D})\mid\psi)$ and $p_2(\mathcal{D}\mid \eta,\psi)=p(\mathcal{D}\mid \eta,\psi)/p_\text{pl}(r(\mathcal{D})\mid\psi)$ at \eqref{eq:two_module_system}, then $p_\text{cut}(\psi\mid r(\mathcal{D}))\propto p_\text{pl}(r(\mathcal{D})\mid\psi)p(\psi)$ and the cut posterior is 
\begin{equation}\label{eq:type2cut}
    p_\text{cut}(\psi,\eta\mid\mathcal{D})=p_\text{cut}(\psi\mid r(\mathcal{D}))p(\eta\mid \psi, \mathcal{D})\,.
\end{equation}

Evaluating~\eqref{eq:pl} is difficult computationally.
To avoid this the authors introduce auxiliary variables $u=(u_1^\top,\ldots,u_n^\top)^\top$, with $u_i=(u_{i1},\ldots,u_{id})^\top$, and use the more tractable extended pseudo likelihood
\begin{equation}\label{eq:epl}
    p_\text{epl}(r(\mathcal{D}),u\mid\psi)=\prod_{i=1}^n c(u_i;\psi)\prod_{j=1}^d \indicator{a_{ij}(\mathcal{D})\leq u_{ij}\leq b_{ij}(\mathcal{D})}\,.
\end{equation}
Integrating \eqref{eq:epl} over $u$ returns the pseudo rank likelihood. They then compute the cut posterior augmented with $u$ given by $p_\text{cut}(\psi,u,\eta\mid\mathcal{D})=p_\text{cut}(\psi,u\mid r(\mathcal{D}))p(\eta\mid \psi, \mathcal{D})$, where $p_\text{cut}(\psi,u\mid r(\mathcal{D}))\propto p_\text{epl}(r(\mathcal{D}),u\mid\psi)p(\psi)$. The marginal in $\psi,\eta$ is the desired cut posterior at~\eqref{eq:type2cut}. 

\subsection{Multiple cuts of marginals}\label{sec:multicut}
The approach of~\cite{SmiYuNotFra2025} allows for only a single cut from one module for all marginals to the copula. However, in many applications some marginal distributions might be more trusted than others, and a researcher may wish to separately specify whether or not to cut each marginal. To do this, we extend their framework as follows.

Let $\delta=(\delta_1,\dots,\delta_d)^\top\in\{0,1\}^d$ be a vector of binary indicators such that $\delta_j=0$ indicates that the $j$th marginal is cut and $\delta_j=1$ that the $j$th marginal is uncut. The set $D(\delta)=\{j; \delta_j=0\}=\{j_1,\ldots,j_r\}$ includes the indices of the marginal distributions that are cut, and $D^c(\delta)=\{j; \delta_j=1\}=\{j_{r+1},\ldots,j_d\}$ those that remain uncut.
Then we can partition the marginal parameters $\eta$ into cut $\eta_{D(\delta)}=(\eta_{j_{1}}^\top,\ldots,\eta_{j_r}^\top)^\top$ and uncut $\eta_{D^c(\delta)}=(\eta_{j_{r+1}}^\top,\ldots,\eta_{j_d}^\top)^\top$ components. 

To define the cut posterior we replace the extended pseudo likelihood at~\eqref{eq:epl} with the following extension:
\begin{align}\label{eq:epl_multiple}
    p_{\text{mpl},\delta}(\mathcal{D},u\mid\psi,\eta_{D^c(\delta)})&=\prod_{i=1}^n c(u_i;\psi)\prod_{j\in D(\delta)} \indicator{a_{ij}(\mathcal{D})\leq u_{ij}\leq b_{ij} (\mathcal{D})}\times\nonumber\\
    &\prod_{j\in D^c(\delta)}\indicator{u_{ij}=F_j(y_{ij};\eta_j)}\,.
\end{align}
Integrating over $u$ gives a pseudo likelihood where the marginals for $D(\delta)$ only use their rank data; see Supporting Information~A. This resulting pseudo likelihood is a mixed density, so that we call it a mixed pseudo likelihood. 
When $\delta=(0,\ldots,0)^\top$, $D^c(\delta)=\emptyset$ and~\eqref{eq:epl_multiple} is exactly equal to~\eqref{eq:epl}, whereas when $\delta=(1,\ldots,1)^\top$, $D(\delta)=\emptyset$ and~\eqref{eq:epl_multiple} is the copula density evaluated at points $u_{ij}=F_j(y_{ij};\eta_j)$ for all $i,j$.

For a fixed vector $\delta$, the joint cut posterior (augmented with $u$) is given by setting 
\begin{align}
p_{\text{cut},\delta}(\psi,u,\eta\mid\mathcal{D})=&p_{\text{mpl},\delta}(\mathcal{D},u\mid\psi,\eta_{D^c(\delta)})p(\psi,\eta_{D^c(\delta)})\times \nonumber\\&\left\{\prod_{i=1}^n \prod_{j\in D^c(\delta)}f_j(y_{ij};\eta_j)\right\}p(\eta_{D(\delta)}\mid \psi,\eta_{D^c(\delta)},\mathcal{D}).\label{eq:jntcutmulti}
\end{align}
Here, the conventional conditional posterior
$p(\eta_{D(\delta)}\mid \psi,\eta_{D^c(\delta)},u,\mathcal{D})=p(\eta_{D(\delta)}\mid \psi,\eta_{D^c(\delta)},\mathcal{D})$ is unaffected by $u$.
The marginal of~\eqref{eq:jntcutmulti} in $(\psi,\eta)$ is the required cut posterior.  

This approach is a natural extension to Section~\ref{sec:SmiYuNotFra2025}, and it agrees with the approach of \citet{SmiYuNotFra2025} when $\delta=(0,\dots,0)^\top$.
For $\delta=(1,\dots,1)^\top$, $D(\delta)=\emptyset$, so that from~\eqref{eq:jntcutmulti}, 
\begin{equation*}p_{\text{cut},1}(\psi,\eta\mid\mathcal{D}) =  \int   p_{\text{cut},1}(\psi,u,\eta\mid\mathcal{D}) du = p(\psi,\eta|\mathcal{D})\,,
\end{equation*}
is the conventional posterior. 

However, conducting inference using this framework has two challenges.
First, selecting the optimal cut $\delta$ can be difficult as this leads to a discrete search problem over $2^d$ potential cuts. Second, the degree of misspecification in each marginal can vary, and a full cut of each misspecified marginal may be disadvantageous depending on the objective of the analysis. 
%First, evaluating~\eqref{eq:jntcutmulti} can be difficult computationally. Second, selecting the optimal cut $\delta$ can be difficult as this leads to a discrete search problem over $2^d$ potential cuts. 
To overcome both challenges, we instead propose a novel way of controlling the information flow between marginals and copula parameters in a continuous manner which we describe below. It allows each marginal
to be either fully cut, fully uncut, or partially cut. By doing so, it produces a continuous relaxation of the hard cut, where the search is no longer over a discrete space of dimension $2^d$, but over the hypercube $[0,1]^d$.

\subsection{Our SMI approach}\label{sec:smi_copula}
Let $\gamma=(\gamma_1,\dots,\gamma_d)^\top$ denote a vector of parameters, $0\leq\gamma_j\leq1$, where $\gamma_j$ controls the influence of the $j$th marginal parameter $\eta_j$ on $\psi$. When $\gamma_j=0$ the influence of the $j$th marginal is fully cut from the copula module. Then, we propose using
the following SMI extended likelihood   
\begin{equation}\label{eq:epl_cont}
    p_{\text{SMI},\gamma}(\mathcal{D},u\mid\psi,\eta)=\prod_{i=1}^n c(u_i;\psi)\prod_{j=1}^d \indicator{a_{ij}(\gamma_j,\eta_j,\mathcal{D})\leq u_{ij}\leq b_{ij}(\gamma_j,\eta_j,\mathcal{D})},
\end{equation}
where the bounds interpolate between the ranks and  $F_j(y_{ij},\eta_j)$:
\begin{align*}
    a_{ij}(\gamma_j,\eta_j,\mathcal{D}) &= \gamma_jF_j(y_{ij};\eta_j)+(1-\gamma_j)\frac{r(y_{ij})-1}{n+1},\\
    b_{ij}(\gamma_j,\eta_j,\mathcal{D}) &= \gamma_jF_j(y_{ij};\eta_j)+(1-\gamma_j)\frac{r(y_{ij})}{n+1}.
\end{align*}
Equation~\eqref{eq:epl_cont} is a continuous relaxation of \eqref{eq:epl_multiple} and exactly equals $p_{\text{mpl},\gamma}$ when $\gamma=\delta\in\{0,1\}^d$. 
In particular, it recovers the extended pseudo likelihood at~\eqref{eq:epl} when $\gamma=(0,\dots,0)^\top$, and when $\gamma=(1,\dots,1)^\top$
\begin{align*}
    p_{\text{SMI},1}(\mathcal{D},u\mid\psi,\eta)=\prod_{i=1}^n c\left(\left(F_1(y_{i1};\eta_1),\dots,F_d(y_{id};\eta_d)\right)\mid\psi\right)
\end{align*}
we recover the copula density term from the parametric likelihood. The full likelihood \eqref{eq:copula} can therefore be written as
\begin{align*}
    p(\mathcal{D}\mid\psi,\eta)&=\prod_{i=1}^n c\left(\left(F_1(y_{i1};\eta_1),\dots,F_d(y_{id};\eta_d)\right)\mid\psi\right)\prod_{j=1}^d f_j(y_{ij};\eta_j)\\&=p_{\text{SMI},1}(\mathcal{D},u\mid\psi,\eta)\prod_{i=1}^n\prod_{j=1}^d f_j(y_{ij};\eta_j).
\end{align*}
Based on this observation we now define a generalized Bayes posterior for fixed $\gamma\in[0,1]^d$ as follows. Similar to the SMI framework of \citet{CarNic2022} we introduce an additional auxiliary parameter $\widetilde{\eta}=(\widetilde{\eta}_1^\top,\ldots,\widetilde{\eta}_d^\top)^\top$ and set 
\begin{equation}\label{eq:powpost}
    p_{\text{SMI},\gamma}(\psi,u,\widetilde{\eta}\mid \mathcal{D})\propto p(\psi,\widetilde{\eta})p_{\text{SMI},\gamma}(\mathcal{D},u\mid\psi,\widetilde{\eta})\prod_{i=1}^n\prod_{j=1}^df_j(y_{ij}\mid\widetilde{\eta}_j).
\end{equation}
This is a posterior derived from a tempered likelihood, where $\gamma_j$ controls the influence of $\widetilde{\eta}_j$ on $\psi$. Combining~\eqref{eq:powpost} with the conditional posterior $p(\eta\mid\psi,\mathcal{D})$ defines an SMI posterior
distribution augmented with both $u$ and $\widetilde{\eta}$: 
\begin{align}\label{eq:SMI_full}
    p_{\text{SMI},\gamma}(\psi,u,\eta,\widetilde{\eta}\mid \mathcal{D})=p_{\text{SMI},\gamma}(\psi,u,\widetilde{\eta}\mid\mathcal{D})p(\eta\mid\psi,\mathcal{D}).
\end{align}
Information from the marginals can influence $\psi$ through the auxiliary parameter $\widetilde\eta$, which can influence the bounds $a_{ij}(\gamma_j,\widetilde{\eta}_j,\mathcal{D})$ and $b_{ij}(\gamma_j,\widetilde{\eta}_j,\mathcal{D})$ in $p_{\text{SMI},\gamma}(\psi,u,\widetilde{\eta}\mid\mathcal{D})$, to an extent that is controlled by $\gamma$. 
Integration over the auxiliary parameters $u$ and $\widetilde{\eta}$ in \eqref{eq:SMI_full} gives the SMP of the copula and marginal parameters: 
\begin{equation}
p_{\text{SMI},\gamma}(\eta,\psi\mid\mathcal{D})=\int\int p_{\text{SMI},\gamma}(\psi,u,\eta,\widetilde{\eta}\mid \mathcal{D}) du d\widetilde{\eta}\,.\label{eq:smijntpost}
\end{equation}

The pseudo posterior at \eqref{eq:powpost} is similar to the power posterior at~\eqref{eq:powerpostCN} in the sense that it uses a $\gamma$-modified likelihood controlling the influence of auxiliary parameters $\widetilde{\eta}$ on the parameter $\psi$ of the trusted module. However, this novel approach to SMI in copula models differs substantively from the SMI framework in \citet{CarNic2022}. Under their power posterior approach, the untrusted module for $\eta_j$ would reduce to the prior for $\gamma_j=0$, whereas in our case the likelihood term for the $j$th marginal, $f_j(y\mid\widetilde{\eta}_j)$, is unaffected by $\gamma_j$. That is, even if $\widetilde{\eta}_j$ is fully cut from the copula module, $\widetilde{\eta}_j$ is still informed by the data. This difference is crucial.  Our structure ensures that, for each $j$, $F_j(y_{ij};\widetilde{\eta}_j)$ is fully informed by the marginal likelihood for all $\gamma$, and the bounds $a_{ij}(\gamma_j,\widetilde{\eta}_j,\mathcal{D})$ and $b_{ij}(\gamma_j,\widetilde{\eta}_j,\mathcal{D})$ in \eqref{eq:epl_cont} are a mixture between the (parameter independent) rank data and the parametric marginals, with mixture weight given by $\gamma_j$.

\subsection{Efficient variational inference}\label{sec:computation}

Exact sampling from $p_{\text{SMI},\gamma}(\psi,u,\eta,\widetilde{\eta}\mid \mathcal{D})$ lends itself to a two step approach, where in a first step a sample from $p_{\text{SMI},\gamma}(\psi,u,\widetilde{\eta}\mid \mathcal{D})$ is generated and then in a second step $\eta\sim p(\eta\mid\psi,\mathcal{D})$ is drawn. However, this is computationally prohibitive because it would involve two nested MCMC samplers, and require an efficient sampler for the high-dimensional density \eqref{eq:powpost}, so that we use an efficient VI approach instead. 
\cite{YuNotSmi2023} and \cite{CarNic2022} have used VI to compute
cut and semi-modular posteriors previously. Here we combine a structured Gaussian variational family with ``stop gradient'' operators \citep{CarNic2022} that allow updating of all variational parameters jointly.

All model parameters $\theta=(\psi^\top,\eta^\top)^\top$ are either unconstrained or transformed monotonically to be so. Similarly, let $\Phi$ denote the standard normal distribution function, then each $u_{ij}$ 
%$u_{ij}=(b_{ij}-a_{ij})\Phi(z_{ij})+a_{ij}$ 
is transformed to an unconstrained value $z_{ij}=\Phi^{-1}\left((u_{ij}-a_{ij})/(b_{ij}-a_{ij})\right)$ for all $i,j$.
A priori, $u_{ij}\sim \UD(a_{ij}(\gamma_i,\eta_j,\mathcal{D}),b_{ij}(\gamma_i,\eta_j,\mathcal{D}))$, so that $u$ depends on $\eta$, whereas each $z_{ij}\sim \ND(0,1)$ is independent of $\eta$ a priori. Thus $z=(z_1^\top,\ldots,z_n^\top)^\top$, with $z_i=(z_{i1},\ldots,z_{id})^\top$, is an attractive
re-parameterization that simplifies the geometry of the prior and posterior. 
The variational approximation of the reparameterized posterior $p_{\text{SMI},\gamma}(\psi,z,\eta,\widetilde{\eta}\mid \mathcal{D})$ is of the form
\begin{align*}
q_\lambda(\psi,z,\eta,\tilde{\eta})=q_{\lambda}(\psi)q_{\lambda}(z)q_\lambda(\eta\mid\psi)q_\lambda(\tilde{\eta}\mid\psi),
\end{align*}
where $q_{\lambda}(\psi), q_\lambda(\eta\mid\psi)$, and $q_\lambda(\tilde{\eta}\mid\psi)$ are Gaussians as follows:
\begin{align*}
    q_{\lambda}(\psi) &= \varphi(\psi;\mu_\psi,T_\psi^{-\top}T_\psi^{-1})\\
    q_\lambda(\eta\mid\psi) &= \varphi(\eta;\mu_\eta-T_\eta^{-\top}T_{\psi\eta}^\top(\psi-\mu_\psi);T_\eta^{-\top}T_\eta^{-1})\\
    q_\lambda(\tilde{\eta}\mid\psi) &= \varphi(\tilde{\eta};\mu_{\tilde{\eta}}-T_{\tilde{\eta}}^{-\top}T_{\psi\tilde{\eta}}^\top(\psi-\mu_\psi);T_{\tilde{\eta}}^{-\top}T_{\tilde{\eta}}^{-1}),
\end{align*}
with $T_\psi$, $T_\eta$, and $T_{\tilde\eta}$ lower triangular matrices  with positive diagonals, and $T_{\psi\tilde{\eta}}$, $T_{\psi\eta}$ unrestricted matrices. This parameterization is chosen so that $q_\lambda(\psi,\eta,\tilde{\eta})$ is jointly Gaussian with $\eta$ and $\tilde{\eta}$ being independent conditional on $\psi$ \citep{TanNot2018}. Following \citet{LoaizaMayaSmith2019}, $q_\lambda(z)=\prod q_\lambda(z_{ij})$ are independent Gaussian approximations with $q_\lambda(z_{ij})=\varphi(z_{ij},\zeta_{ij},\omega_{ij}^2)$. The full vector of variational parameters to be optimized is thus $$\lambda=\left(\zeta^\top,\omega^\top,\mu_\psi^\top,\mu_\eta^\top,\mu_{\tilde{\eta}}^\top,\text{vec}(T_\psi),\text{vec}(T_\eta),\text{vec}(T_{\tilde{\eta}}),\text{vech}(T_{\psi\eta}),\text{vech}(T_{\psi\tilde{\eta}})\right)^\top,$$
where $\text{vec}(A)$ and $\text{vech}(A)$ denote the vectorization and half-vectorization of a matrix $A$, respectively. 

Following~\cite{YuNotSmi2023}, the variational parameters $\lambda$ can be learned in two VI steps. First, the sub-vector parameterizing $q_\lambda(\psi,z,\tilde{\eta})$ approximating $p_{\text{SMI},\gamma}(\psi,u,\widetilde{\eta}\mid \mathcal{D})$ can be learned. Then, in a second step the variational parameters $(\mu_{\eta}^\top,\text{vech}(T_{\eta})^\top,\text{vec}(T_{\psi,\eta})^\top)^\top$ can be trained while keeping the remaining parameters fixed. %This two step approach ensures that information flow from $\eta$ onto $\psi$ is fully stopped. 
However, this two step approach involves two separate stochastic optimizations. Instead, we update the full vector of variational parameters $\lambda$ in a single joint step enabling end-to-end training applying stop gradient operators as now described. 
From a computational point of view, a stop gradient operator prevents gradients from flowing backwards during automatic differentiation.%, and has been popularized by the implementation in TensorFlow~\citep{tensorflow}.

Denote a copy of $\lambda$ which takes the same value as $\lambda$ but is treated as a constant when deriving the gradient $\nabla_\lambda$ as 
``$\cancel{\lambda}$''. To employ the reparameterization trick, a draw from $q_\lambda(\psi,z,\eta,\tilde{\eta})$ is obtained by transforming independent standard normal draws as follows:
\begin{itemize}[label={}]
\item Draw $\varepsilon_{ij}^{(z)}\sim\ND(0,1)$ and set $z_{ij}=\zeta_{ij}+\omega_{ij}\varepsilon_{ij}^{(z)}$ for $i=1,\dots,n$, $j=1,\dots,d$.
\item  Draw $\varepsilon^{(\psi)}\sim\ND(0,I)$ and set $\psi=\mu_\psi+T_\psi^{-\top}\varepsilon^{(\psi)}$.
\item Draw $\varepsilon^{(\eta)}\sim\ND(0,I)$ and set $\eta=\mu_\eta-T_\eta^{-\top}T_{\psi\eta}^\top(\cancel{\psi}-\cancel{\mu_\psi})+T_\eta^{-\top}\varepsilon^{(\eta)}$.
\item Draw $\varepsilon^{(\tilde{\eta})}\sim\ND(0,I)$ and set $\tilde{\eta}=\mu_{\tilde{\eta}}-T_{\tilde{\eta}}^{-\top}T_{\psi\tilde{\eta}}^\top(\psi-\mu_\psi)+T_{\tilde{\eta}}^{-\top}\varepsilon^{(\tilde{\eta})}$.
\end{itemize}

Writing $h_\gamma(\psi,u,\widetilde{\eta})$ for the right hand side in \eqref{eq:powpost}, we note that $h_\gamma(\psi,u,\widetilde{\eta})$ is the kernel of $p_{\text{SMI},\gamma}(\psi,u,\widetilde{\eta}\mid\mathcal{D})$, and $h_1(\cancel{\psi},\cancel{u},\eta)$ is the kernel of the conventional conditional posterior $p(\eta\mid\psi,\mathcal{D})$. Hence, the ELBO for the variational optimization can be expressed as 
\begin{equation} \label{eq:elbo_stopgradient}
    \mathcal{L}(\lambda)=\mathbb{E}_{q_\lambda}\left[\log h_\gamma(\psi,u,\tilde{\eta})+\log h_1(\cancel{\psi},\cancel{u},\eta)-\log q_\lambda(\psi,z,\eta,\tilde{\eta})\right].
\end{equation}

An unbiased estimate of the gradient of $\mathcal{L}(\lambda)$ is thus given as $\widehat{\nabla_\lambda\mathcal{L}\left(\lambda\right)}=\nabla_\lambda \left(\log h_\gamma(\psi,u,\tilde{\eta})+\log h_1(\cancel{\psi},\cancel{u},\eta)-\log q_\lambda(\psi,z,\eta,\tilde{\eta})\right)$. Due to the way the stop-gradient operator is employed in generating draws from $q_\lambda$ and in the evaluation of \eqref{eq:elbo_stopgradient} this updating scheme is equivalent to the two-step approach described above. 

\subsection{Selecting the influence parameter}
So far, we have considered $\gamma$ as a fixed vector. 
However, selecting $\gamma$ is challenging, with similar issues also being encountered for the power posterior in the original SMI of~\cite{CarNic2022}. 
The theoretical results presented in Section~\ref{sec:theory} demonstrate that the (asymptotic) behaviour of the SMP depends directly on $\gamma$.
Section~\ref{sec:simulation} gives an empirical example, in which selecting $\gamma$ to correctly match the misspecification of the marginals in the data generating process does not necessarily lead to an optimal joint SMP. 

Each influence parameter $\gamma_j$ only directly controls the influence of the $j$th marginal on the copula function module, and its influence on the other marginals is indirect. 
%with dependence between parameters of different marginals only indirectly controlled through $\gamma$. 
For this reason, we suggest choosing the optimal vector $\gamma^*$ using an external utility function $u(\gamma)$. 
The utility function may depend on additional data sources not considered for training and is largely determined by the objectives of the empirical application. 
For example, if the goal is prediction, $u(\gamma)$ can be a predictive metric evaluating the sharpness and precision of the copula-based forecasting model. 
Alternatively, if the main interest is estimation of the dependence structure, $u(\gamma)$ can be the expected log-copula likelihood.

%After specifying $u(\gamma)$, the optimal vector of influence parameters is given as $\gamma^\ast=\arg\max_\gamma u(\gamma)$. 
Even for moderate $d$, optimizing $\gamma$ directly is infeasible, and so we propose to use BO \citep{garnett23}.
BO is a commonly employed strategy for efficiently optimizing black-box functions lacking gradients. BO builds a probabilistic surrogate of the objective $u(\cdot)$ by training a Gaussian Process on observed values and uses that surrogate to choose the most informative next point to evaluate balancing both prediction and learned uncertainty. Standard software for BO is be readily applied to our variational framework, and we use BoTorch~\citep{botorch}.

\section{Theoretical behavior of SMI for copulas} \label{sec:theory}
This section discusses the theoretical behavior of the SMP $p_\smi(\eta,\psi\mid\mathcal{D})$ developed in Section \ref{sec:smi_copula}. While the discussion pertains to copula models as presented here, many of the points raised
apply more generally to SMPs constructed in the fashion of~\cite{CarNic2020}.

 \subsection{$\gamma$-dependent concentration results}
 Our SMPs dependence on $\gamma$ ultimately implies that the point onto which the posterior concentrates must be viewed as $\gamma$-dependent. In contrast to the usual learning rate encountered in generalized Bayesian inference,  the value of $\gamma$ directly influences not just posterior width but location (i.e., concentration) as well. This suggests that one should optimize over $\gamma$ to determine which values deliver accurate inferences and predictions. 

 Given that the SMP contains the fully cut posterior at \eqref{eq:type2cut} as a special case, $\gamma_j=0$ for all $j$, and given that the theoretical behavior of standard posteriors, $\gamma_j=1$ for all $j$, are well-known, we restrict our theoretical analysis to values where $\gamma_j\in[0,1]$ and where $\gamma_j\not\in\{0,1\}$ for at least some $j$. 

Define the joint vector of parameters $\varphi=(\psi^\top,\widetilde{\eta}^\top)^\top\in\Psi\times\mathcal{E}$, the following negative SMI ``log-pseudo-likelihood'' and its corresponding limit counterpart as
$$
M_n(\varphi;\gamma):=-\log \int p_{\smi,\gamma}(\mathcal{D},u\mid \psi,\widetilde\eta)\dt u,\mbox{ and } \mathcal{M}(\varphi;\gamma):=\lim_{n\rightarrow\infty}\E [M_n(\varphi;\gamma)/n].
$$Further, define the $\gamma$-dependent pseudo-true value as
$$
\varphi_0(\gamma):=\argmin_{\varphi\in\Psi\times\mathcal{E}}\mathcal{M}(\varphi;\gamma).
$$Likewise, for 
$
\ell_n(\eta,\psi)=\prod_ic(u_{i1},\dots,u_{id};\psi)\prod_jf(y_{ij}\mid\eta)$, and $\mathcal{L}_{\infty}(\eta,\psi)=\lim_n\E[\ell_n(\eta,\psi)]/n
$ define $$
\eta_0(\gamma):=\argmin_{\eta\in\mathcal{E}}\mathcal{L}_\infty\{\eta,\psi_0(\gamma)\}.$$ 

%Since $p_{\smi,\gamma}(\varphi\mid \mathcal{D})$  can be represented as a generalized posterior based on the loss $M_n^\gamma(\varphi)$, whose corresponding population optimizer $\varphi_0(\gamma)$ changes with $\gamma$, it is clear that the theoretical behavior of $p_{\smi,\gamma}(\varphi\mid \mathcal{D})$ will depend on the value of $\gamma$. T
The following result clarifies the SMP's dependence on $\gamma$: posterior concentration occurs but towards a $\gamma$-dependent quantity. 
\begin{theorem}\label{thm:phi_conc}
Under Assumptions 1-4 in Supporting Information~B, for a positive sequence $r_n\rightarrow0$, $M_n$ large enough, some $C>0$ and $\alpha>0$, with probability converging to one,
$$
\int \mathds{1}\{\|\psi-\psi_0(\gamma)\|>CM_n^\alpha r_n^\alpha\}p_{\smi,\gamma}(\psi\mid\mathcal{D})\dt\psi\rightarrow0.
$$If, in addition, Assumption 5 in Supporting Information~B is satisfied, then
$$
\int \int \mathds{1}\{\|\eta-\eta_0(\gamma)\|>C\sqrt{M_n r_n}\}p_{\smi,\gamma}(\eta,\psi\mid\mathcal{D})\dt\psi\dt\eta\rightarrow0.
$$with probability converging to one.
\end{theorem}

Theorem~\ref{thm:phi_conc} demonstrates that the posteriors $p_{\smi,\gamma}(\psi\mid\mathcal{D})$ and $p_{\smi,\gamma}(\eta\mid \mathcal{D})$ concentrate onto a $\gamma$-dependent minimizer; their dependence on the modularity parameter $\gamma$ then implies that results beyond posterior concentration are unlikely to be illuminating since the posterior depends on $\gamma$ in a nonlinear manner. In particular, the existing results on the asymptotic shape of cut posteriors in \cite{frazier2025cutting} are unlikely to be satisfied for this SMP due to the aforementioned nonlinear dependence. Further, we note that the differing rate of concentrations in Theorem~\ref{thm:phi_conc} is a consequence of the proof technique we have employed, rather an a genuine difference: since the posterior $p_{\smi,\gamma}(\eta\mid \mathcal{D})$ depends on the conditional posterior $p(\eta\mid\psi,\mathcal{D})$, we must control the behavior of $p(\eta\mid\psi,\mathcal{D})$ uniformly over $\psi$, resulting in the slower rate of convergence. It is indeed possible to refine this result to obtain the faster rater of convergence, under additional assumptions. However, since the goal of Theorem~\ref{thm:phi_conc} is to formally demonstrate the impact of the modularity parameter $\gamma$ on  $p_{\smi,\gamma}(\eta,\psi\mid \mathcal{D})$, the sub-optimal rate is sufficient for this purpose.

%The form of the SMP means that results beyond posterior concentration derived in Theorem~\ref{thm:phi_conc} are unlikely to be enlightening due to its nonlinear dependence on $\gamma$. In particular, the existing theoretical results on the asymptotic shape of cut posteriors in \cite{frazier2025cutting} are unlikely to be satisfied for the SMP, and additional regularity conditions are required to obtain results on the uncertainty quantification for this class of SMPs. Due to the technical nature of such constructions, we propose to study these issues  separately and in more generality. 

\subsection{Interpretation of results}
The proposed SMP behaves differently to power-posterior derived SMPs. Because we have multiple values of $\gamma$, if $\gamma_j\in(0,1)$ for at least one $j$, the likelihood for the $j$-th marginal parameter $\widetilde{\eta}_j$ is still informed by data. Consequently, our theoretical results can be interpreted as showing that the SMP defines a posterior process $\{\gamma\in[0,1]^d: p_{\smi,\gamma}(\psi,\eta\mid\mathcal{D})\}$, where the points onto which $p_{\smi,\gamma}(\psi,\eta\mid\mathcal{D})$ concentrate explicitly depend on the values of $\gamma$. 

This behavior means that, even though the value of $\gamma$ is treated as a hyper-parameter in $p_{\smi,\gamma}$, its impact on the resulting inferences is substantially different from other hyper-parameters in generalized Bayesian inference such as the learning rate. In contrast, it is well-known that the choice of learning rate does not impact the point onto which the posterior concentrates, at least in regular models; see \cite{syring2023gibbs} and \cite{mclatchie2025predictive} for discussion of this point. 

The dependence of the SMP on $\gamma$ suggests that it may be possible to learn an optimal value of $\gamma$ according to some measure of accuracy. Because each value of $\gamma\in[0,1]^d$ indexes a specific generalized posterior, and thus some loss function, we could in-principle select $\gamma$ via the loss function selection method proposed by~\cite{jewson2022general} where uniform priors (on 0 to 1) are adopted for each component $\gamma_j$. However, in contrast to the analysis of these authors, there is no reason to suspect that the value of $\gamma$, and hence $\varphi_0(\gamma)$, will be uniquely identified from the data. 

To demonstrate this fact, assume that the marginal models $F_j\left(y_{i j} ; \eta_j\right)$ are well-specified. Then, writing $r(y_{ij})/(n+1)=\frac{n}{n+1}F_n(y_{ij})$ we see that, for $n$ large, 
\begin{flalign*}
a_{i j}\left(\gamma_j, \eta^0_j, \mathcal{D}\right)&=\gamma_j F_j\left(y_{i j} ; \eta^0_j\right)+\left(1-\gamma_j\right) \frac{r\left(y_{i j}\right)-1}{n+1}=F(y_{ij};{\eta}_j^0)-\frac{(1-\gamma_j)}{n+1}+o(n),\\
b_{i j}\left(\gamma_j, \eta^0_j, \mathcal{D}\right)&=\gamma_j F_j\left(y_{i j} ; \eta_j\right)+\left(1-\gamma_j\right) \frac{r\left(y_{i j}\right)}{n+1}= \left(\frac{n}{n+1}\right)F(y_{ij};{\eta}_j^0)+o(n),
\end{flalign*}
since $\frac{n}{n+1}F_n(y_{ij})\approx F(y_{ij};{\eta}_j^0)$ under correct model specification. Hence, if the marginals are well-specified, the value of $\gamma_j$ is irrelevant, and unidentifiable asymptotically. Conversely, if the model is misspecified, then the value of $\gamma_j$ may have a non-negligible impact: define $F_0$ to be the true distribution and let $\Delta(\gamma_j)=\gamma_j\{F(y_{ij};\eta_j)-F_0(y_{ij})\}$, then we see that  
\begin{flalign*}
a_{i j}\left(\gamma_j, \eta_j, \mathcal{D}\right)&=\Delta(\gamma_j)+F_0(y_{ij})\left(\frac{n}{n+1}\right)-\frac{(1-\gamma_j)}{n+1}+o(n),\\
b_{i j}\left(\gamma_j, \eta_j, \mathcal{D}\right)&=\Delta(\gamma_j)+F_0(y_{ij})\left(\frac{n}{n+1}\right)+o(n),
\end{flalign*}which depends on $\gamma_j$, and will therefore influence the values of $\psi,\eta$ onto which the posterior concentrates. 

%We therefore see that there is no reason to suspect the existence of a unique solution to the limit optimization problem:
%$$
%\min_{\gamma\in[0,1]^d,\varphi\in\Psi}\left\{\mathcal{M}_\infty^\gamma(\varphi)+\mathcal{L}_\infty^\gamma(\eta)\right\}.
%$$Consequently, it is likely that the approach suggested in \cite{jewson2022general} to identify and consistently estimate hyper-parameters in generalized Bayesian loss functions will fail in this setting, since no unique optima may exist. 

This does not mean that one cannot attempt to target some value in the equivalence class of optima. Herein, we have attempted to select such a value of $\gamma$, and ultimately the value of $\varphi_\star(\gamma)$, through BO. We note here that the theoretical analysis of such a procedure is complicated by the lack of identification and so we leave a detailed theoretical analysis of the optimal value of $\gamma$, and thus the point onto which the SMP concentrates to further research. 

Before finishing this section, we note that many of the points raised by our theory are not unique to our SMP but are more generally true, under different conditions, for the SMP of \cite{CarNic2020}. That is, many of the techniques we have used to derive results for our SMP are directly applicable to the SMP of \cite{CarNic2020}, and in this way minor modifications of our results deliver concentration results for the SMP approach of \cite{CarNic2020}. Consequently, similar to our SMP, the SMP of \cite{CarNic2020} will also concentrate onto a value of the parameters that is $\gamma$-dependent. We do note, however, that this will not be the case for the SMP of \cite{FraNot2025}, which builds a linear pool between two different posteriors, and ultimately has different behavior to the SMI considered here. That being said, the SMP of \cite{FraNot2025} does not allow for dimension-by-dimension cutting as the modularity is applied at the level of the posterior, which is much less flexible than the approach proposed herein.

\section{Simulations} \label{sec:simulation}
To illustrate how the SMP mediates information flow between modules we consider a bivariate copula model where the dependence is correctly specified and only one marginal is misspecified. 
We generate replicate datasets from the data generating process (DGP), and vary $\gamma=(\gamma_1,\gamma_2)^\top$ over a grid with the following three goals: (i) to illustrate how reducing the influence of the misspecified marginal can improve estimation of the dependence structure, (ii) to characterize the trade-offs the SMP introduces for inference on the correctly specified marginal that is coupled to the misspecified marginal through the copula, and (iii) to show how $\gamma$ can be selected using a simple utility function in practice. 
An additional simulation study considering a high-dimensional example with $d=15$ marginals is given in Supporting Information~C.

\noindent\textbf{Data generating process:} We generate a total of $m=100$ datasets, each with $n=1,000$ observations, from a bivariate Gumbel copula model. The copula is parameterized through Kendall's tau, with true value $\tau^\ast=0.7$. The first marginal with density $f_1^*$ is a log-normal distribution with location $\mu^\ast=-1$, and scale $\sigma^\ast=1$. The second marginal $f_2^*$ is gamma distributed with shape $\alpha^\ast=2$ and rate $\beta^\ast=1$. 

\noindent\textbf{Statistical model:} To each dataset we fit a Gumbel copula model with log-normal marginals with densities $f_1, f_2$, 
so that the copula and the first marginal are correctly specified, while the second marginal is misspecified. 
The model parameters are $\tau$ for the Gumbel copula, as well as $(\mu_j,\sigma^2_j)$, $j=1,2$, for the marginals. 
Writing $\psi=\log\frac{\tau}{1-\tau}$, $\eta_j=\left(\mu_j,\log\sigma_j\right)^\top$, $j=1,2$, $\eta=(\eta_1^\top,\eta_2^\top)^\top$ gives the vector of unconstrained parameters $\theta=\left(\psi,\eta^\top\right)^\top$. We adopt the following weakly informative priors $\tau\sim\UD(0,1)$, $\mu_1,\mu_2\sim\ND(0,100^2)$, and $\sigma^2_1,\sigma^2_2\sim\HND(0,100^2)$, where $\HND$ is 
a half-normal distribution.

\noindent\textbf{Performance metrics:} We evaluate accuracy for each module separately. To measure the fit to the copula, we use the mean squared error for Kendall's tau $\tau(\psi)=(1+\exp(-\psi))^{-1}$, integrating over the SMP for $\psi$,
\begin{equation*}
    \mathcal{L}_\text{cop}=\int (\tau(\psi)-\tau^*)^2 p_{\text{SMI},\gamma}(\psi\mid\mathcal{D})d\psi.
\end{equation*}
We evaluate the fit of each marginal by computing their Kullback-Leibler divergences, integrating over the SMP for $\eta_j$,
\begin{equation*}
    \mathcal{L}_j=\int\KL(f_j^*(y)\|f_j(y;\eta_j))p_{\text{SMI},\gamma}(\eta_j\mid\mathcal{D})d\eta_j,\qquad j=1,2,
\end{equation*}
with $\KL(f_j^*(y)\|f_j(y;\eta_j))=\int (\log f_j^*(y) - \log f_j(y;\eta_j))f_j^*(y)dy$. The integrals are numerically approximated using Monte Carlo samples, and smaller values are preferred.

\noindent\textbf{Results:} Figures~\ref{fig:sim_grid}(A-C) show average values for $\mathcal{L}_1$, $\mathcal{L}_2$, and $\mathcal{L}_\text{cop}$ across the $100$ replicates. In each panel,
$\gamma=(\gamma_1,\gamma_2)^\top$ is varied over a mesh on $[0,1]^2$ with a grid size of $0.1$. 
Both $\mathcal{L}_2$ and $\mathcal{L}_\text{cop}$ exhibit a similar structure: small values for $\gamma_2$ and values close to $1$ for $\gamma_1$ improve the estimation of the copula and the second marginal.
Thus, inference is improved for these modules when influence from the second but not from the first marginal is reduced, matching the model misspecification. 
In contrast, $\mathcal{L}_1$ is lowest when $\gamma_2=1$ (i.e. the misspecified marginal is uncut). 

Notice that $\gamma$ controls the influence of $\eta$ on $\psi$, but only 
indirectly controls the information flow between $\eta_1$ and $\eta_2$ through the copula. This leads to the following observation.
Under the conventional posterior (i.e. $\gamma=(1,1)^\top$) $\tau$ is underestimated; for example, the maximum (conventional) a-posteriori estimator for $\tau$ has an average bias of $-0.0113$ over all repetitions. 
This underestimation restricts the influence of the parameters of the misspecified marginal $\eta_2$ on $\eta_1$. But in SMI, selecting $\gamma$ to improve estimation of the copula (i.e. $\psi$) increases the information flow from the misspecified second marginal to the well-specified first marginal, decreasing estimation accuracy for $\eta_1$. 

Nevertheless, because the first marginal is well-specified and its rank estimator is consistent, the value of $\gamma_1$ becomes less relevant for large $n$, as shown in Section~\ref{sec:theory}.
Figure~\ref{fig:sim_grid} illustrates this behavior where, for small $\gamma_2$, the value of $\gamma_1$ has only very limited influence on all three performance metrics. In particular, $\mathcal{L}_\text{cop}$ is almost flat on $[0,1]\times[0,0.5]$, suggesting that an optimal $\gamma$ may not be uniquely identified. 

\begin{figure}[tb]
    \centering
    \includegraphics[width=0.8\linewidth,keepaspectratio]{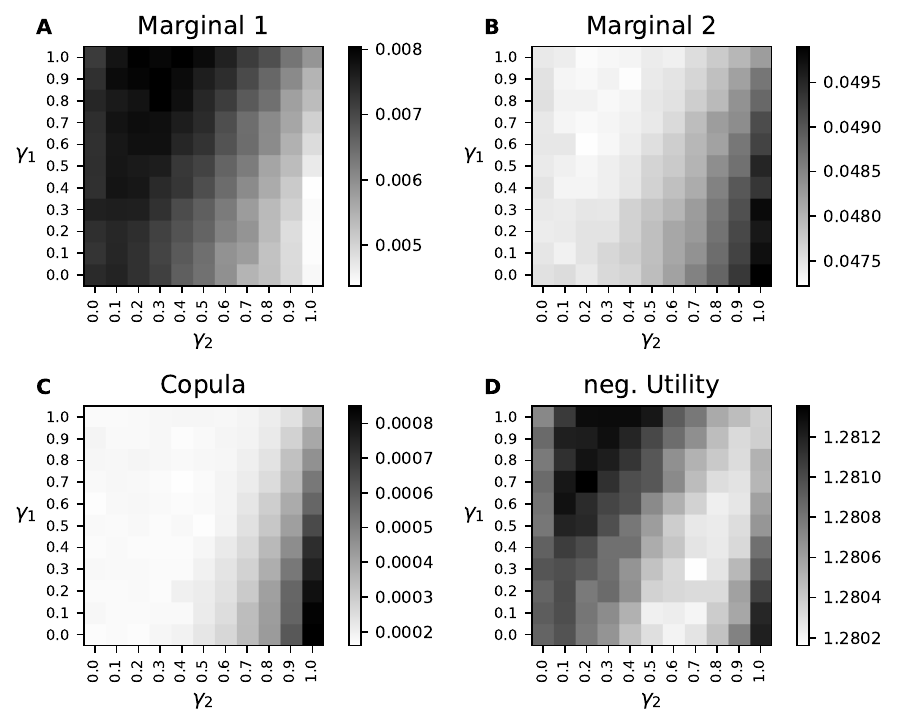}
    \caption{\small Simulations. Average values for $\mathcal{L}_1$ (Panel A), $\mathcal{L}_2$ (Panel B), $\mathcal{L}_\text{cop}$ (Panel C), and $-u(\gamma)$ (Panel D) across $m=100$ repetitions for different combinations of $\gamma=(\gamma_1,\gamma_2)^\top$. Smaller values indicated by lighter colour are preferred. The conventional posterior $\gamma=(1,1)^\top$ corresponds to the upper right corner, while the lower left corner is the fully cut posterior $\gamma=(0,0)^\top$.}
    \label{fig:sim_grid} 
\end{figure}

\noindent\textbf{Selecting $\gamma$:} The results presented above illustrate the difficulties in selecting $\gamma$. Even if knowledge of the misspecification of the marginals is available, choosing $\gamma$ to match this misspecification can lead to unforeseen consequences in the estimation. In this example, the first marginal is correctly specified while the second marginal is misspecified, indicating a choice of $\gamma=(1,0)^\top$. This choice improves the estimation of the dependence structure and the misspecified marginal, but estimation accuracy of the first marginal deteriorates.

As shown in Section~\ref{sec:theory} the behavior of the SMP, and in particular the point on which the posterior concentrates asymptotically, depends directly on $\gamma$. Here, we consider the utility function
\begin{align}\label{eq:u_sim}
u(\gamma)=\int \sum_{i=1}^n \log f(y_{i};\psi,\eta)p_{\text{SMI},\gamma}(\psi,\eta\mid\mathcal{D})d\psi d\eta,
\end{align}
as an external metric to select an optimal value $\gamma^*$, and balance the trade-off between the marginal and copula fits described above. Figure~\ref{fig:sim_grid}(D) shows average values for $-u(\gamma)$ across the $100$ repetitions. SMPs with $\gamma\in(0,1)^2$ outperform the conventional and three possible cut posteriors with (i.e. where $\gamma\in\{0,1\}^2$) by this metric.

\section{Stock market volatility and bond yields}\label{sec:yieldseg}

\subsection{Problem description and data} 
There is growing evidence that volatility in the stock market and bond yields are dependent~\citep{Jubinski2012,MoeSorJFE2022}. However, the form of this dependence has not been explored previously. We do so here using a copula model with a skew-normal (SN) copula that allows for asymmetric dependence between variable pairs. Parametric marginals are preferred over nonparametric ones because focus is on the tail behavior of the variables, and we use a ``Sinh–Arcsinh'' (SAS) distribution~\citep{jones2009sinh}. This distribution is a four-parameter family constructed from a transformation of a Gaussian base distribution that models skewness and kurtosis flexibly. However, one or more of these marginals may be misspecified, so that computing the SMP is attractive here. 
 
For bond yields we consider the effective yield on domestic U.S. publicly issued debt graded AAA ($Y_{\text{AAA},t}$) and BBB ($Y_{\text{BBB},t}$) by Standard \& Poors. The former grade is the least risky, while the latter is more risky and the most liquid. Daily values from 1 January 2022 to 30 June 2025 were sourced from the Federal Reserve Bank of St. Louis, which corresponds to the post-pandemic period when U.S. interest rates rose quickly from historical lows. 
We consider the dependence of the yields with U.S. equity market volatility
measured using the average daily values of the VIX ($Y_{\text{VIX},t}$) sourced from Yahoo Finance. All variables are transformed to the real line using their logarithms, and we fit the copula model to the $d=6$ dimensional vector $Y_t=(Y_{\text{VIX},t},Y_{\text{AAA},t},Y_{\text{BBB},t},Y_{\text{VIX},t-1},Y_{\text{AAA},t-1},Y_{\text{BBB},t-1})^\top$. Our goal is to gain insights into the temporal and cross-sectional dependence structure. 

%\textbf{Data and copula model}  Let $y_{\text{VIX},t}$ denote the VIX at day $t$. The VIX is an index of overall market volatility and the data was sourced from Yahoo finance. $y_{\text{AAA},t}$ and $y_{\text{BBB},t}$ represent the effective yield publicly issued in the US domestic market given investment grade rating AAA and BBB respectively. This data is made publicly available by the Federal Reserve Bank of St. Louis. After a log-transformation, we jointly model the $d=6$ dimensional vector $y_t=(y_{\text{VIX},t},y_{\text{AAA},t},y_{\text{BBB},t},y_{\text{VIX},t-1},y_{\text{AAA},t-1},y_{\text{BBB},t-1})^\top$ over the time period from 01 January 2022 to 30 June 2025. This time period is selected as it corresponds to a stable market period after the COVID pandemic. Our goal is to gain insights into the temporal and cross-sectional dependence structure using a skew-Gaussian copula model. 

\subsection{Copula model}
The SN copula is the implicit copula of the SN distribution of~\cite{AzzVal1996}; see~\citet{smith2023implicit} for an overview of implicit copulas. Let $X=(X_1,\ldots,X_d)^\top\sim F_X$ be distributed SN with location zero, scale matrix $\Omega$, and skew parameter $\alpha\in\mathbb{R}^d$. Then it has density 
%$C(u)=F_X(F^{-1}_{X_1}(u_1),\dots,F^{-1}_{X_d}(u_d))$, where $F_X$ denotes the multivariate skew-normal distribution of~\cite{AzzVal1996} with marginals %$F_{X_j}$, $j=1,\dots,d$. 
%This copula allows for asymmetric tail dependence making it an attractive choice. 
%A random vector $X\sim F_X$ has density
$f_X(x;\Omega;\alpha)=2\varphi(x;0,\Omega)\Phi(\alpha^\top x)$,
where $x=(x_1,\ldots,x_d)^\top$.
%$\varphi(x;\mu,\Sigma)$ denotes the density of an $N(\mu,\Sigma)$ distribution, and $\Phi$ is the standard normal distribution function. 
%The vector
%$\alpha\in\mathbb{R}^d$ is a skewness parameter and $\Omega\in\mathbb{R}^{d\times d}$ is a covariance matrix. 
The SN distribution is closed under marginalization, with $F_{X_j}$ a zero mean univariate SN distribution with skew parameter $\bar{\alpha}_j$ that is a function of 
$\{\alpha,\Omega\}$. Then, $X$ has copula function $C(u)=F_X(F^{-1}_{X_1}(u_1),\dots,F^{-1}_{X_d}(u_d))$ with density
\begin{equation*}
    c(u\mid \Omega,\alpha)=\frac{f_X(x;\Omega,\alpha)}{\prod_{j=1}^d f_{X_j}(x_j;\Bar{\alpha}_j)},
\end{equation*}
where $x_j=F^{-1}_{X_j}(u_j)$ and $f_{X_j}=\frac{\partial}{\partial x_j}F_{X_j}$.
Setting $\Omega$ to a correlation matrix identifies the copula parameters $\{\alpha,\Omega\}$ and enforces $F_{X_j}$ to have unit scale for $j=1,\ldots,d$.

We assume that the three series $(Y_{\text{VIX},t},Y_{\text{AAA},t},Y_{\text{BBB},t})$ follow a 3-dimensional 
stationary stochastic process. 
Because $Y_t$ contains lagged values of these three series, this assumption further restricts the parameter space, with $\Omega$ a block symmetric matrix, $\alpha$ constrained, and the marginals for $Y_{j,t}$ and $Y_{j,t-1}$ being the same for each $j\in\{\text{VIX},\text{AAA},\text{BBB}\}$. Full specification of the copula parameterization under this stationary assumption is given in Supporting Information D.1, where the 
dimensions of $\psi$, $\eta$ and $\gamma$ are 15, 12 and 3, respectively.
Finally, when evaluating the copula density and function, calculating $F_{X_j}^{-1}$ is a key computation, and we do this using the spline interpolation approach outlined in~\cite{SmiMan2018}.

%Since the marginals do not match a known parameteric distribution, we use the flexible $4$ parameter sinh-arcsinh normal distribution with distribution function 
%\begin{equation*}
%    F(y)=\Phi\left(\sinh\left(\delta+\kappa\sinh^{-1}\left(\frac{y-\mu}{\sigma}\right)\right)\right),
%\end{equation*}
%where $\Phi$ denotes the CDF of the standard normal distribution and $\sinh(\cdot)$ the hyperbolic sine. $\mu$ controls the location, $\sigma$ the scale, $\delta$ the skewness, and $\kappa$ the kurtosis of the distribution. Parameters for the marginal for $Y_{j,t}$ and $Y_{j,t-1}$, $j\in\{\text{VIX},\text{AAA},\text{BBB}\}$, are shared so that $\eta$ is $12$-dimensional.

\subsection{Empirical results}
\textbf{Selecting $\gamma$:} We consider the utility function
\begin{equation}\label{eq:selection_criterion}
    u(\gamma) = \int \sum_{t=1}^T
    \sum_{j\in\{\text{VIX},\text{AAA},\text{BBB}\}} 
    \log f(y_{j,t};\eta_j)p_{\text{SMI},\gamma}(\eta\mid\mathcal{D})d\eta,
\end{equation}
which we optimize using BO. The fully cut model with $\gamma=(0,0,0)^\top$ has a higher utility value than the conventional posterior where $\gamma=(1,1,1)^\top$. The optimal vector of influence parameters maximizing \eqref{eq:selection_criterion} is $\gamma^\ast=(1.00,0.61,0.00)^\top$.

\textbf{Marginal fit:} Figure~\ref{fig:app_marginals} shows the estimated marginal densities computed at the posterior means of the conventional, cut and SMI posteriors, overlaid with histograms of the data. A SAS distribution captures the VIX marginal well, but not those of the AAA and BBB yields. 
%The histograms for the two bond yields are highly left-skewed and exhibit strong kurtosis in the left tail. 
The estimated marginal density under the conventional posterior does not fully capture the modes of these two marginals, which are better captured by the cut posterior and SMP. The marginal of the BBB yield variable has the strongest misspecification, followed by that of the AAA yield variable, which is consistent with the optimal $\gamma^\ast$ value. 

\begin{figure}[tbh!]
    \centering
    \includegraphics[width=0.8\linewidth,keepaspectratio]{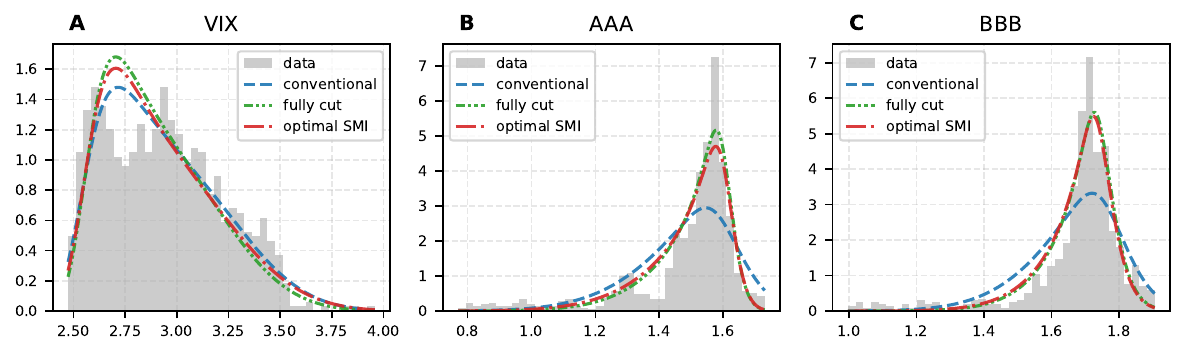}
    \caption{\small Plots of the estimated marginal densities in the bond yields example, for the VIX (Panel~A), AAA yield (Panel~B), and BBB yield (Panel~C). Histograms of the data are given in gray, and the density estimates are evaluated at the posterior means $\widehat{\eta}$ for the conventional posterior (blue dashed line), the fully cut posterior, $\gamma=(0,0,0)^\top$, (green dotted line), and the optimal copula-SMI posterior (red dashed-dot line).
    }
    \label{fig:app_marginals} 
\end{figure}

\begin{figure}[!htbp]
    \centering
    \includegraphics[width=0.8\linewidth]{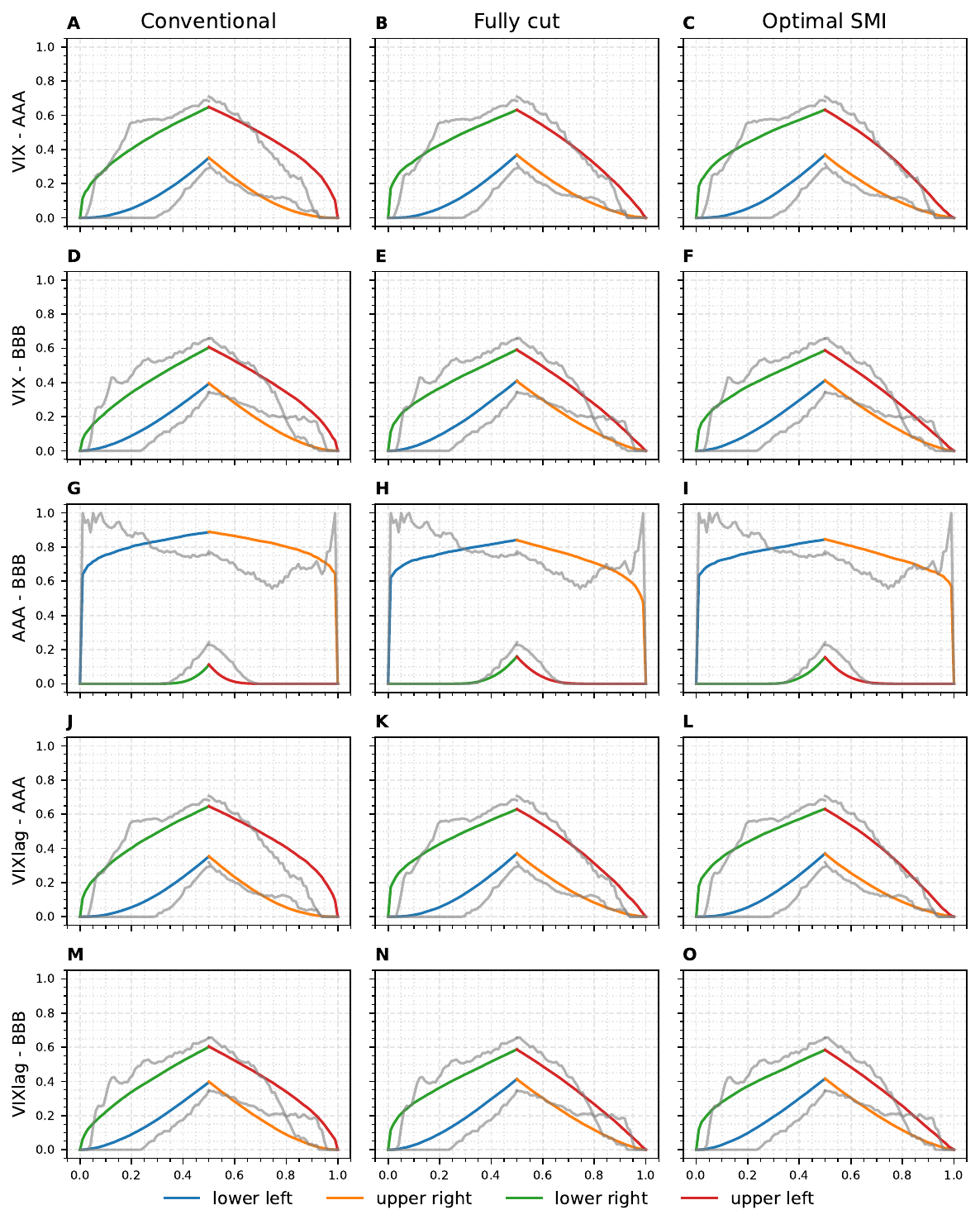}
    \caption{\small Bond yields example. $\rho_{LL}(\zeta)$ (blue), $\rho_{UR}(1-\zeta)$ (orange), $\rho_{LR}(\zeta)$ (green), and $\rho_{UL}(1-\zeta)$ (red) versus $\zeta$ for the five ordered pairs VIX-AAA, VIX-BBB, AAA-BBB, VIXlag-AAA and VIXlag-BBB (rows) for the conventional (left), fully cut (middle) and SMP (right) means of $\psi$.}
    \label{fig:quantile_dependence}
\end{figure}

\textbf{Dependencies:} The main objective of the analysis is estimate the dependence structure between the variables, which we summarize using pairwise metrics. 
Let $U=(U_1,\dots,U_d)\sim C(u;\psi)$, we define for a pair $(U_j,U_l)$, $j\not=l$, and a given quantile $\zeta\in(0,0.5)$ dependence in each 
of the quadrants as: $\rho_\text{LL}(\zeta;U_j,U_l)=\mathbb{P}(U_l\leq\zeta\mid U_j\leq \zeta)$, $\rho_\text{UR}(\zeta;U_j,U_l)=\mathbb{P}(U_l>1-\zeta\mid U_j>1-\zeta)$, $\rho_\text{LR}(\zeta;U_j,U_l)=\mathbb{P}(U_l\leq\zeta\mid U_j>1-\zeta)$, and $\rho_\text{UL}(\zeta;U_j,U_l)=\mathbb{P}(U_l>1-\zeta\mid U_j\leq\zeta)$,
%\begin{itemize}[itemsep=0pt]
%    \item[] Lower Left: $\rho_\text{LL}(\zeta;U_j,U_l)=\mathbb{P}(U_l\leq\zeta\mid U_j\leq \zeta)$, 
%    \item[] Upper Right: $\rho_\text{UR}(\zeta;U_j,U_l)=\mathbb{P}(U_l>1-\zeta\mid U_j>1-\zeta)$,
%    \item[] Lower Right: $\rho_\text{LR}(\zeta;U_j,U_l)=\mathbb{P}(U_l\leq\zeta\mid U_j>1-\zeta)$, 
%    \item[] Upper Left: $\rho_\text{UL}(\zeta;U_j,U_l)=\mathbb{P}(U_l>1-\zeta\mid U_j\leq\zeta)$,
%\end{itemize}
which are computed from the bivariate marginal copula of $(U_j,U_l)$. Note that the variable pairs are strictly ordered, and switching their order gives  the identities $\rho_\text{LL}(\zeta;U_j,U_l)=\rho_\text{UR}(\zeta;U_l,U_j)$ and 
$\rho_\text{LR}(\zeta;U_j,U_l)=\rho_\text{UL}(\zeta;U_l,U_j)$.
Following~\cite{OhPat2013}, Figure~\ref{fig:quantile_dependence} presents ``quantile dependence plots'' for the five ordered pairs VIX-AAA, VIX-BBB, AAA-BBB, VIXlag-AAA and VIXlag-BBB. These visualize asymmetric dependence along the major diagonal by plotting  
 $\rho_\text{LL}(\zeta)$ and $\rho_\text{UR}(1-\zeta)$ (blue and orange lines), and along the minor diagonal by plotting $\rho_\text{LR}(\zeta)$ and $\rho_\text{UL}(1-\zeta)$ against $\zeta$ (green and red lines). Estimates are given for the conventional, cut and SMI posterior means of $\psi$ in the three columns. In each panel, the empirical (i.e. nonparametric) equivalent is also given in gray for comparison. Major diagonal dependence between AAA and BBB are strong for all estimators, reflecting the fact that bond yields move together very tightly in response market events. However, the estimated dependence structure between the yields and the VIX (contemporaneous or lagged) differs greatly between conventional and cut/SMI posteriors. The conventional posterior suggests dependence is symmetric, whereas both the cut posterior and SMP suggest strong asymmetry which is more consistent with the empirical equivalents. The asymmetry is most apparent in the minor diagonal, as overall dependence between the VIX and yields is negative.  

Following~\cite{DenSmiMan2025}, pairwise asymmetry in the major and minor diagonals for the bivariate copula of $(U_j,U_l)$ can be measured at quantile $\zeta$ by  
\begin{eqnarray*}
\Delta_\text{Major}(\zeta;U_j,U_l)&=&\rho_\text{UR}(\zeta;U_j,U_l)-\rho_\text{LL}(\zeta;U_j,U_l)\,,\mbox{ and }\\
    \Delta_\text{Minor}(\zeta;U_j,U_l)&=&\rho_\text{UL}(\zeta;U_j,U_l)-\rho_\text{LR}(\zeta;U_j,U_l)\,.
\end{eqnarray*}
Switching the order of the variables gives the identities $\Delta_\text{Major}(\zeta;U_j,U_l)=-\Delta_\text{Major}(\zeta;U_l,U_j)$ and $\Delta_\text{Minor}(\zeta;U_j,U_l)=-\Delta_\text{Minor}(\zeta;U_l,U_j)$.
Figure~\ref{fig:heatmaps_asymmetry} plots the posterior means of these at quantile $\zeta=0.1$ for all variable pairs in the form of heatmaps. Results are not given for the conventional posterior because they are either exactly or very close to zero for all variable pairs. However, there is a clear difference between those for the cut and SMI posteriors, with the latter indicating strong positive asymmetric dependence between the VIX (contemporanous or lagged) and AAA and BBB yields. This is more economically intuitive, as volatility spikes typically coincide with risk-off episodes that induce flight-to-quality and nonlinear repricing of credit risk. It is also consistent with existing empirical evidence that equity market volatility is more strongly linked to bond yields in stress states than in tranquil periods~\citep{Jubinski2012}.

\begin{figure}[htb]
    \centering
    \includegraphics[width=0.8\linewidth]{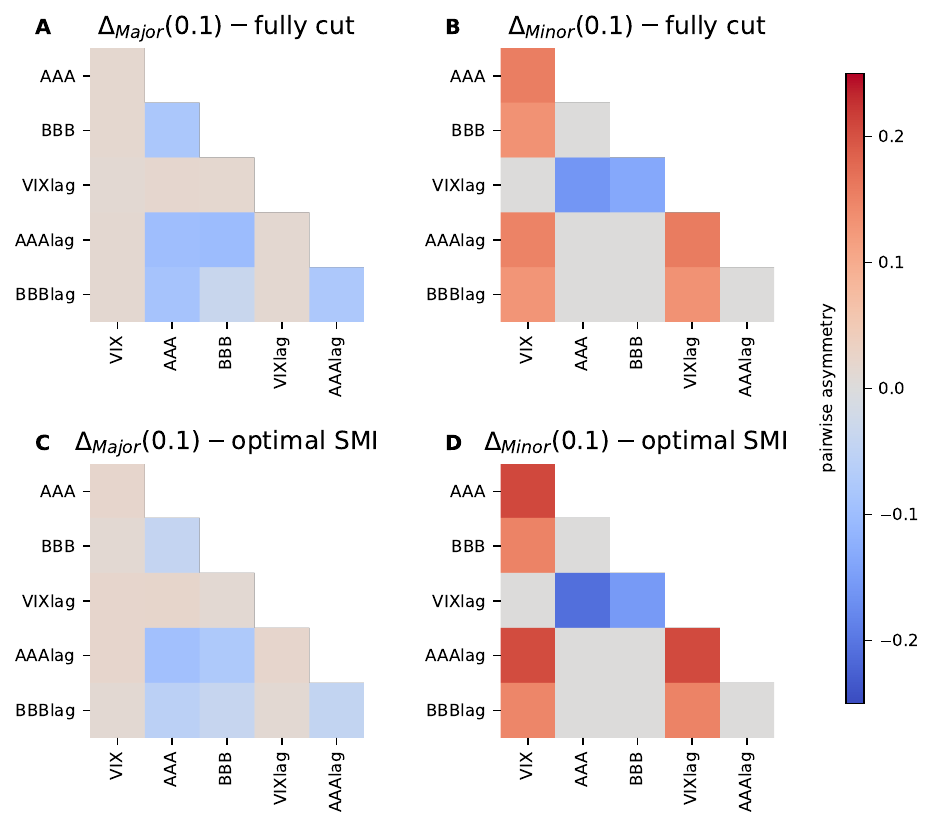}
    \caption{\small Bond yields example. Heatmaps for the pairwise asymmetries $\Delta_\text{Major}(0.1)$ (left column), and $\Delta_\text{Minor}(0.1)$ (right column) for fully cut (top row) and SMP (bottom row) means of $\psi$. Each metric is for the  variable order (x-axis, y-axis); for example, top the value is for the ordered pair (VIX,AAA). Thus, in Panels~B and~D, $\Delta_\text{Minor}(\zeta;U_{\text{AAA}},U_{\text{VIXlag}})<0$ which implies 
    $\Delta_\text{Minor}(\zeta;U_{\text{VIXlag}},U_{\text{AAA}})>0$; similarly for the pair (BBB,VIXlag) and its switched order (VIXlag,BBB).}
    \label{fig:heatmaps_asymmetry}
\end{figure}

Further empirical results, including copula parameter estimates and Kendall's tau values, are given for all three posteriors in Supporting Information~D.2.

\section{Conclusion}\label{sec:conc}
We have developed a novel SMI framework for copula models, designed for settings where inferences about copula parameters must be protected from corruption by misspecified marginal distributions, and where the degree of misspecification varies across marginals. Our approach treats each marginal as its own module with a marginal-specific influence parameter, allowing the framework to adapt to the degree of misspecification by controlling how strongly each marginal shapes the SMP. We provide efficient variational computational methods, a principled strategy for selecting influence parameters via Bayesian optimization of a utility function, and theoretical validation in the form of a posterior concentration result. The practical value of the methodology is demonstrated through simulations and a substantive real application examining the relationship between bond yields and a volatility index, where the SMP most clearly reveals the asymmetric dependence structure compared to the fully cut and conventional posterior.

There are many directions for future work and we mention one of them.  Our paper has focused on treating each marginal as its own module, rather than treating all marginals as a single module.  It is also possible to consider decomposing some copula functions into multiple distinct modules; for example, with a vine copula.  When different copula components are misspecified to varying degrees, the SMI framework could then be used to limit the influence of misspecified copula modules on inference for parameters elsewhere in the model, offering a more flexible approach to robust dependence modelling.

\FloatBarrier

%\section*{References}
%\setlength{\bibsep}{0pt plus 0.3ex}
\spacingset{1.0}
\bibliography{bib}
\spacingset{2.0}
\appendix
\section*{Appendix}

\section{Mixed Pseudo Likelihood}\label{app:mpl}
Equation~(9) in Section~3.2 gives a
pseudo likelihood for the multiple cut case that is extended with the auxilary variables $u$. To integrate over these auxilary variables, partition $u_i$ into uncut $u_{i,D^c(\delta)}=\{u_{ij};j\in D^c(\delta)\}$ and cut 
$u_{i,D(\delta)}=\{u_{ij};j\in D(\delta)\}$ components, and decompose the copula density into
$c(u_i;\psi)=c_{D^c}(u_{i,D^c(\delta)};\psi) c_{D|D^c}(u_{i,D(\delta)}|u_{i,D^c(\delta)};\psi)$ where $c_{D^c}$ is the marginal copula density of $u_{i,D^c(\delta)}$, and $c_{D|D^c}$ is the conditional density of $u_{i,D(\delta)}$ given $u_{i,D^c(\delta)}$. Then if the indices $D(\delta)=\{j_1,j_2,\ldots,j_r\}$,
\begin{eqnarray*}
 p_{\text{mpl},\delta}(\mathcal{D}\mid\psi,\eta_{D^c(\delta)}) &= &\int p_{\text{mpl},\delta}(\mathcal{D},u\mid\psi,\eta_{D^c(\delta)}) du\\
 &= &\prod_{i=1}^n c_{D^c}(u_{i,D^c(\delta)};\psi)\Delta_{a_{i,j_1}(\mathcal{D})}^{b_{i,j_1}(\mathcal{D})}\cdots \Delta_{a_{i,j_r}(\mathcal{D})}^{b_{i,j_r}(\mathcal{D})}C_{D|D^c}\left(u_{i,D(\delta)}|u_{i,D^c(\delta)}\right)
\end{eqnarray*}
with $C_{D|D^c}$ the conditional distribution function, and $u_{ij}=F_j(y_{ij};\eta_j)$ for $j\in D^c(\delta)$. This is a mixed density for the  rank values for the marginals
indexed by $D(\delta)$ and continuous observations for the marginals indexed by $D^c(\delta)$.

Note that the cut posterior fits the two module case by setting
$p_1(\mathcal{D}|\psi,\eta_{D^c(\delta)})=p_{\text{mpl},\delta}(\mathcal{D}|\psi,\eta_{D^c(\delta)})\prod_{j\in D^c(\delta)}f_j(y_{ij};\eta_j)$ and 
$p_2(\mathcal{D}|\psi,\eta)=p(\mathcal{D}|\psi,\eta)/p_1(\mathcal{D}|\psi,\eta_{D^c(\delta)})$.

\section{Theory}\label{sec:theory_app}
This Section contains a discussion on the assumptions, and proofs for the results stated in Section~4 of the manuscript. 

\subsection{Assumptions and Proofs of Main results}
Recall
$$
p_{\text{SMI},\gamma}(\psi\mid\mathcal{D})=\int\int\int p_{\text{SMI},\gamma}(\psi,u,\eta,\widetilde{\eta}\mid \mathcal{D}) du d\widetilde{\eta}d\eta
$$ where 
\begin{align*}
    p_{\text{SMI},\gamma}(\psi,u,\eta,\widetilde{\eta}\mid \mathcal{D})=p_{\text{SMI},\gamma}(\psi,u,\widetilde{\eta}\mid\mathcal{D})p(\eta\mid\psi,\mathcal{D}). 
\end{align*}
%We maintain the following assumption on the vector of influence parameters.
Further, we have that 
\begin{flalign*}
p_{\text{SMI},\gamma}(\psi,u,\widetilde{\eta}\mid\mathcal{D})&\propto p_{\text{SMI},\gamma}(\mathcal{D},u\mid\psi,\widetilde\eta)p(\psi,\widetilde{\eta})\prod_{i=1}^n\prod_{j=1}^df_j(y_{ij}\mid\widetilde{\eta}_j).\\
    p_{\text{SMI},\gamma}(\mathcal{D},u\mid\psi,\eta)&=\prod_{i=1}^n c(u_i;\psi)\prod_{j=1}^d \indicator{a_{ij}(\gamma_j,\eta_j,\mathcal{D})\leq u_{ij}\leq b_{ij}(\gamma_j,\eta_j,\mathcal{D})},
\end{flalign*}
where 
\begin{align*}
    a_{ij}(\gamma_j,\eta_j,\mathcal{D}) &= \gamma_jF_j(y_{ij};\eta_j)+(1-\gamma_j)\frac{r(y_{ij})-1}{n+1},\\
    b_{ij}(\gamma_j,\eta_j,\mathcal{D}) &= \gamma_jF_j(y_{ij};\eta_j)+(1-\gamma_j)\frac{r(y_{ij})}{n+1}.
\end{align*}
Define the SMI posterior, marginal of $u$,
$$
p_{\text{SMI},\gamma}(\psi,\widetilde\eta\mid\mathcal{D})=\frac{\int p_{\text{SMI},\gamma}(\mathcal{D},u\mid\psi,\widetilde\eta)p(\psi,\widetilde{\eta})\prod_{i=1}^n\prod_{j=1}^df_j(y_{ij}\mid\widetilde{\eta}_j)\dt u }{\int\int\int  p_{\text{SMI},\gamma}(\mathcal{D},u\mid\psi,\widetilde\eta)p(\psi,\widetilde{\eta})\prod_{i=1}^n\prod_{j=1}^df_j(y_{ij}\mid\widetilde{\eta}_j)\dt u \dt\widetilde\eta \dt\psi}
$$
However, when $\gamma_j\in[0,1)$ the integral reduces to
\begin{flalign*}
\int p_{\text{SMI},\gamma}(\mathcal{D},u\mid\psi,\widetilde\eta)\dt u&=\prod_{i=1}^{n}\Delta_{a_{i1}(\gamma_1,\widetilde\eta_1,\mathcal{D})}^{b_{i1}(\gamma_1,\widetilde\eta_1,\mathcal{D})}\dots \Delta_{a_{id}(\gamma_d,\widetilde\eta_d,\mathcal{D})}^{b_{id}(\gamma_d,\widetilde\eta_d,\mathcal{D})}C(u_i;\psi)%\\&=\prod_{i=1}^{n}\Delta^\gamma_i(\widetilde\eta,\psi)
\\&=:\Delta_n(\widetilde\eta,\psi;\gamma).
\end{flalign*}

This allows us to write the posterior $p_{\text{SMI},\gamma}(\psi\mid\mathcal{D})$ in a more convenient form: for $f_n(\widetilde{\eta}):=\prod_{i=1}^n\prod_{j=1}^df_j(y_{ij}\mid\widetilde{\eta}_j)$, we have 
$$
p_{\text{SMI},\gamma}(\psi,\widetilde{\eta}\mid\mathcal{D})=\frac{ \Delta_n(\widetilde\eta,\psi;\gamma)p(\psi,{\eta})f_n(\widetilde{\eta})}{\int\int\Delta_n(\eta,\psi;\gamma)p(\psi,\widetilde{\eta})f_n(\widetilde{\eta})\dt \widetilde\eta \dt\psi},
$$which defines the $\psi$-marginal SMP as
$$
p_{\text{SMI},\gamma}(\psi\mid\mathcal{D})=\frac{\int \Delta_n(\widetilde\eta,\psi;\gamma)p(\psi,\widetilde{\eta})f_n(\widetilde{\eta})\dt\widetilde\eta}{\int\int\Delta_n(\widetilde\eta,\psi;\gamma)p(\psi,{\eta})f_n(\widetilde{\eta})\dt \widetilde\eta \dt\psi}.
$$

Note that, although $p_{\text{SMI},\gamma}(\psi\mid\mathcal{D})$ does not have the form of a Gibbs measure,  $p_{\text{SMI},\gamma}(\psi,\widetilde\eta\mid\mathcal{D})$ does in fact have a Gibbs form. In particular, abusing notation and now letting $\varphi=(\psi^\top,\widetilde{\eta}^\top)^\top$, defining 
$
\MD_n(\varphi;\gamma)=-[\log \Delta_n(\widetilde\eta,\psi;\gamma)+\log f_n(\widetilde{\eta})]/n,
$ we have that 
\begin{equation}\label{eq:gibbs_smi}
p_{\text{SMI},\gamma}(\varphi\mid\mathcal{D})=\frac{\exp\{-n\MD_n(\varphi;\gamma)\}p(\varphi)}{\int \exp\{-n\MD_n(\varphi;\gamma)\}p(\varphi)\dt\varphi}\equiv \frac{\exp\{-n[\MD_n(\varphi;\gamma)-\MD_n(\varphi_0(\gamma);\gamma)]\}p(\varphi)}{\int \exp\{-n[\MD_n(\varphi;\gamma)-\MD_n(\varphi_0(\gamma);\gamma)]\}p(\varphi)\dt\varphi}.
\end{equation}

We maintain the following condition that stipulates the support of the influence parameter $\gamma$. 
\begin{assumption}\label{ass:gammas}For all $j=1,\dots,d$,  $\gamma_j\in[0,1]$, with $\gamma_j\not\in\{0,1\}$ for at least one $j$. 
\end{assumption}
In addition to the requirement on $\gamma$ maintained in Assumption \ref{ass:gammas}, we require the following regularity conditions. 

\begin{assumption}\label{ass:loss}
(i) For any $\varphi\in\Phi:=\Psi\times\mathcal{E}$, the function $\varphi\mapsto \MD(\varphi;\gamma):=\E\MD_n(\varphi;\gamma)$ exists. 
(ii) For some $\varepsilon>0$, and any $\gamma$ as in Assumption \ref{ass:gammas}, there exist $\gamma\mapsto \varphi_0(\gamma)$ such that, for all $\delta>0$, and some $\delta'>0$,
$$
\inf_{\varphi:\|\varphi-\varphi_0(\gamma)\|>\delta}\MD(\varphi;\gamma)>\MD\{\varphi_0(\gamma);\gamma\}+\delta'. 
$$(iii) There exist $C>0$, $\alpha>0$, and $U\subset \Phi$ containing $\varphi_0(\gamma)$ such that, for all $\varphi\in U$, $\|\varphi-\varphi_0(\gamma)\|\le C[\MD(\varphi;\gamma)-\MD\{\varphi_0(\gamma);\gamma\}]^\alpha$.
\end{assumption}
\begin{assumption}\label{ass:entropy}
There exists a positive sequence $r_n\downarrow0$ as $n\rightarrow\infty$ such that, for any $\gamma$ as in Assumption \ref{ass:gammas}, $r_n\le Cn^{-1/2}$ and
$
\E \left[\sup_{\varphi\in\Phi}\left|\MD_n(\varphi;\gamma)-\E[\MD_n(\varphi;\gamma)]\right|\right]\le O(r_n).
$
\end{assumption}

\begin{assumption}\label{ass:prior_mass}
Let $\epsilon>0$ with $n\epsilon\ge1$. For $\mathcal{A}_{\epsilon} := \{\varphi\in\Phi: |\MD(\varphi;\gamma)-\MD(\varphi_0(\gamma);\gamma)|\le \epsilon\}$,  the prior satisfies $\int_{\Phi}\mathds{1}\{\varphi\in \mathcal{A}_{\epsilon}\} p(\varphi) \dt\varphi \ge e^{-n\epsilon}$.
% $\Pi(\mathcal{A}_{\epsilon_n})\ge e^{-n\epsilon_n}$.
\end{assumption}

Assumption \ref{ass:loss}(i)-(ii) requires that a minimizer of the loss $\MD(\varphi;\gamma)$ exists for any $\gamma$ at which we wish to compute the SMP; while Assumption \ref{ass:loss}(iii) is a type of Holder-continuity that we must maintain to link fluctuations in the loss near the minimizer, to differences in the parameter space. Assumption \ref{ass:entropy} is a high-level condition that requires the empirical process between the sample and population loss converges at some rate; sufficient conditions for this are discussed in Section \ref{sec:primitive}. Assumption \ref{ass:prior_mass} is a standard prior mass condition used in most applications of generalized Bayesian inference. 

Theorem 1 in the main text is a direct consequence of Theorems \ref{thm:phi_conc} and \ref{thm:eta_conc} that follow.

\begin{theorem}\label{thm:phi_conc}
Under Assumption \ref{ass:gammas}-\ref{ass:prior_mass},  for $n$ and $M>0$ large enough, 
$$
\E\int \mathds{1}\left\{\|\psi-\psi_0(\gamma)\|>CM^\alpha r_n^\alpha\right\}p_{\mathrm{SMI},\gamma}(\psi\mid\mathcal{D})\dt\psi\leq 1/M.
$$  
\end{theorem}
\begin{proof}
Let us simplify notation and drop dependence on $\gamma$ for now. Similarly, let us write $q_n(\varphi)=p_{\mathrm{SMI},\gamma}(\varphi\mid\mathcal{D})$, and $q_n(\psi)=\int q_n(\varphi)\dt\widetilde\eta$. The proof proceeds through two main steps. 
\begin{enumerate}
    \item First, we derive a bound on the probability of $\{\MD(\varphi)-\MD(\varphi_0)\}>M\epsilon$ under $q_n(\varphi)$. 
    \item We then use the variational definition of $q_n(\varphi)$ to derive a rate for this bound. 
    \item We then transfer this bound to $\|\varphi-\varphi_0\|$.
    \item A marginalization argument is then used to arrive at a similar bound for $q_n(\psi)$.
\end{enumerate}

\noindent\textbf{Step 1.} Fix $\epsilon>0$. From Markov's inequality,
\begin{flalign}
 Q_n\left[\{\MD(\varphi)-\MD(\varphi_0)\}>{(M\epsilon)}\right]\leq \frac{1}{(M\epsilon)}\int \{\MD(\varphi)-\MD(\varphi_0)\}q_n(\varphi)\dt\varphi.\label{eq:markov}
\end{flalign}
Rewrite and bound the integral term in \eqref{eq:markov} as 
\begin{flalign}
&\int \{\MD(\varphi)-\MD(\varphi_0)\}q_n(\varphi)\dt\varphi\nonumber\\&=\int\left[\left\{\MD(\varphi)-\MD_n(\varphi)\right\}+\left\{\MD_n(\varphi)-\MD_n(\varphi_0)\right\}+\left\{\MD_n(\varphi_0)-\MD(\varphi_0)]\right\}   \right]q_n(\varphi)\dt\varphi\nonumber 
\\&\le 2\sup_{\varphi\in\Phi}|\MD(\varphi)-\MD_n(\varphi)|+\int \left\{\MD_n(\varphi)-\MD_n(\varphi_0)\right\}q_n(\varphi)\dt\varphi\label{eq:bound_two}
\end{flalign}Substituting \eqref{eq:bound_two} into \eqref{eq:markov}, taking expectations and using Assumption \ref{ass:entropy} we obtain 
\begin{flalign}
&\E  Q_n\left[\{\MD(\varphi)-\MD(\varphi_0)\}>M\epsilon\right]\nonumber\\&\leq   \frac{2}{(M\epsilon)}\E\left[\sup_{\varphi\in\Phi}|\left\{\MD(\varphi)-\MD_n(\varphi)\right\}|\right]+ \frac{1}{(M\epsilon)}\E\int \left\{\MD_n(\varphi)-\MD_n(\varphi_0)\right\}q_n(\varphi)\dt\varphi\nonumber\\&\leq \frac{r_n}{(M\epsilon)}+\frac{1}{(M\epsilon)}\E\int \left\{\MD_n(\varphi)-\MD_n(\varphi_0)\right\}q_n(\varphi)\dt\varphi\label{eq:end_s1}.
\end{flalign}

\noindent\textbf{Step 2.} We now use the variational definition of $q_n(\varphi)$ to bound the second term in \eqref{eq:end_s1}. Since $q_n(\varphi):=p_{\text{SMI},\gamma}(\varphi\mid\mathcal{D})$ as in \eqref{eq:gibbs_smi}, $q_n(\varphi)$ can be equivalently defined as the solution to the following variational optimization problem:
\begin{flalign}\label{eq:gibbs_variation}
q_n(\varphi):=\argmin_{q\in\mathcal{P}(\Phi)}\left\{\int [\MD_n(\varphi)-\MD_n(\varphi_0)]q(\varphi)\dt\varphi+\frac{\KL(q\|p)}{n}\right\}. 
\end{flalign}In particular, we see that, from \eqref{eq:gibbs_variation}, 
\begin{flalign*}
\int [\MD_n(\varphi)-\MD_n(\varphi_0)]q_n(\varphi)\dt\varphi&\le \int [\MD_n(\varphi)-\MD_n(\varphi_0)]q_n(\varphi)\dt\varphi+\frac{\KL(q_n\|p)}{n}\nonumber\\&=\inf_{q\in\mathcal{P}(\Phi)}\left\{\int [\MD_n(\varphi)-\MD_n(\varphi_0)]q(\varphi)\dt\varphi+\frac{\KL(q\|p)}{n}\right\},
\end{flalign*}where the first inequality uses the fact that $\KL(q_n\|p)>0$, and the equality follows from \eqref{eq:gibbs_variation}. Apply expectations to both sides of the above, and use Jensen's inequality to obtain:
\begin{flalign}
   \E \int \{\MD_n(\varphi)-\MD_n(\varphi_0)\}q_n(\varphi)\dt\varphi&\le \E\inf_{q\in\mathcal{P}(\Phi)}\left\{\int \{\MD_n(\varphi)-\MD_n(\varphi_0)\}q(\varphi)\dt\varphi+\frac{\KL(q\|p)}{n}\right\}\nonumber\\&\le\inf_{q\in\mathcal{P}(\Phi)}\left\{\int \{\MD(\varphi)-\MD(\varphi_0)\}q(\varphi)\dt\varphi+\frac{\KL(q\|p)}{n}\right\}.\label{eq:inter_step}
\end{flalign}
Recalling the set $\mathcal{A}_\epsilon$ in Assumption \ref{ass:prior_mass}, we define the restricted prior measure 
$$
\rho_n(\varphi)=\begin{cases}\frac{p(\varphi)}{\int_{\mathcal{A}_\epsilon}p(\varphi)\dt\varphi}&\text{ if }\varphi\in\mathcal{A}_\epsilon\\0&\text{ else}    
\end{cases}.
$$
Since the above is true for all $q\in\mathcal{P}(\Phi)$, take $q=\rho_n$ in \eqref{eq:inter_step} and apply Assumption \ref{ass:prior_mass} to see that 
\begin{flalign*}
   \E \int \{\MD_n(\varphi)-\MD_n(\varphi_0)\}q_n(\varphi)\dt\varphi&\le\inf_{q\in\mathcal{P}(\Phi)}\left\{\int \{\MD(\varphi)-\MD(\varphi_0)\}q(\varphi)\dt\varphi+\frac{\KL(q\|p)}{n}\right\}\\&\le \int \{\MD(\varphi)-\MD(\varphi_0)\}\rho_n(\varphi)\dt\varphi+\frac{\KL(\rho_n\|p)}{n}\\&\le \epsilon +\frac{\KL(\rho_n\|p)}{n},
\end{flalign*}where the second term follows since for $\varphi\in\mathcal{A}_\epsilon$, we have $\{\MD(\varphi)-\MD(\varphi_0)\}\le \epsilon$ by construction. For the second term, note that by Assumption \ref{ass:prior_mass} we have that 
$$
\KL(\rho_n\|p)=-\log\int_{\mathcal{A}_\epsilon}p(\varphi)\dt\varphi\le n\epsilon.
$$Putting this together,  we arrive at the bound
\begin{flalign}
 \E \int \{\MD_n(\varphi)-\MD_n(\varphi_0)\}q_n(\varphi)\dt\varphi\le 2\epsilon.\label{eq:bound_s2}
\end{flalign}

Applying the bound in \eqref{eq:bound_s2} into \eqref{eq:end_s1} we arrive at 
\begin{flalign*}
 \E  Q_n\left[\{\MD(\varphi)-\MD(\varphi_0)\}>M\epsilon\right]&\leq \frac{1}{M}\frac{r_n}{\epsilon}+\frac{2}{M}=\frac{r_n+2\epsilon}{M\epsilon}
\end{flalign*}Take $\epsilon=r_n$ so that  $\frac{r_n+2\epsilon}{M\epsilon}\le 4/M$ implies that $\E  Q_n\left[\{\MD(\varphi)-\MD(\varphi_0)\}>M r_n\right]\leq 4/M\asymp  1/M$.
\\

\noindent\textbf{Step 3.}  Let $\delta>0$ be such that $\{\varphi\in\Phi:\|\varphi-\varphi_0\|\le \delta\}\subset U$, where $U$ is defined in Assumption \ref{ass:loss}(iii). By Assumption \ref{ass:loss}(ii), there exists $\delta'>0$ such that $\|\varphi-\varphi_0\|>\delta$ implies that $\MD(\varphi)>\MD(\varphi_0)+\delta'$. Let $n$ be large enough so that $\delta'>Mr_n$, which implies that 
$$
\left\{\varphi\in\Phi: \{\MD(\varphi)-\MD(\varphi_0)\}\le Mr_n\right\}\subset U.
$$
However, from \textbf{Step 3} we know that 
$$
\E Q_n\left\{\MD(\varphi)-\MD(\varphi_0)\le Mr_n\right\}\gtrsim 1- \frac{1}{M},
$$and applying Assumption \ref{ass:loss}(iii) implies that 
$$
\E Q_n\left\{\|\varphi-\varphi_0\|\le CM^\alpha r_n^\alpha\right\}\gtrsim 1- \frac{1}{M}. 
$$Negating the event implies, for $M$ large enough,   
$
\E Q_n\left\{\|\varphi-\varphi_0\|> CM^\alpha r_n^\alpha\right\}\lesssim \frac{1}{M}. 
$
\\

\noindent\textbf{Step 4.} We can now transfer the bound obtained in \textbf{Step 3} using the definition of $q_n(\psi)$. In particular, fix $\delta>0$ and $\delta'$ as in \textbf{Step 3}. For this choice, if $\|\varphi-\varphi_0\|\le CM^\alpha r^\alpha_n$, then $\|\psi-\psi_0\|\le CM^\alpha r^\alpha_n$ since $\psi$ is a subvector of $\varphi$. Therefore, we see that 
$$
\{\psi:\|\psi-\psi_0\|>CM^\alpha r^\alpha_n\}\subseteq\{\varphi:\|\varphi-\varphi\|>CM^\alpha r^\alpha_n\}.
$$Hence, we have that 
\begin{flalign*}
\E Q_n[\|\psi-\psi_0\|>CM^\alpha r^\alpha_n]&:=\E\int\int \mathds{1}[\|\psi-\psi_0\|>CM^\alpha r^\alpha_n]q_n(\psi,\eta)\dt\eta\\&\le \E Q_n[\|\varphi-\varphi_0\|>CM^\alpha r^\alpha_n]\\&\lesssim 1/M,
\end{flalign*} where the last line follows by applying the conclusion of \textbf{Step 3}.
\end{proof}

To derive a concentration rate for $\eta$ based on the concentration in $\psi$, we require an appropriate link between the two parameters: define the fixed-$\psi$ limit optimizer
$$
\eta_\star(\psi):=\argmin_{\eta\in\mathcal{E}}\mathcal{L}_\infty(\eta,\psi),
$$which coincides with $\eta_0(\gamma)=\eta_\star\{\psi_0(\gamma)\}$ when $\psi=\psi_0(\gamma)$. We require the following additional condition. 
\begin{assumption}
For some $\epsilon>0$, and all $\|\psi-\psi_0(\gamma)\|\le\epsilon$, the following are satisfied.
\begin{enumerate}
    \item $\eta_\star(\psi)$ exists and for some $\delta>0$, we have $\{\eta: \|\eta-\eta_\star(\psi)\|\le\delta\}\subseteq\mathcal{E}$.
    \item If $\|\eta-\eta_\star(\psi)\|\le\delta$, for some $0<\underline{c}\le\overline{c}<\infty$,  $$\underline{c}\|\eta-\eta_\star(\psi)\|^2\le [\mathcal{L}_\infty(\eta,\psi)-\mathcal{L}_\infty\{\eta_\star(\psi),\psi\}]\le \overline{c}\|\eta-\eta_\star(\psi)\|^2$$
    \item Uniformly over $\mathcal{E}$, $p(\eta\mid\psi)\le \overline{c}_p<\infty$, and if $\|\eta-\eta_\star(\psi)\|\le\delta$, then $0<\underline{c}_p\le p(\eta\mid\psi).$
\end{enumerate}
\end{assumption}
\begin{theorem}\label{thm:eta_conc}
Under Assumptions 1-5, for some $M_n\rightarrow\infty$ as $n\rightarrow\infty$
$$
\int \int \mathds{1}\{\|\eta-\eta_0(\gamma)\|> \sqrt{M_nr_n}\}p_{\smi,\gamma}(\eta,\psi\mid\mathcal{D})\dt\psi\dt\eta
$$converges to zero in probability.
\end{theorem}
\begin{proof}
 First, recall that $p_{\smi,\gamma}(\eta,\psi\mid\mathcal{D})=p_{\smi,\gamma}(\psi\mid\mathcal{D})p(\eta\mid\psi,\mathcal{D})$. % and that
% $$
% p(\eta\mid\psi,\mathcal{D})\propto p(\eta\mid\psi)p_{\smi,\gamma=1}(\mathcal{D}\mid \eta,\psi).
% $$ 
%
We again drop terms dependence on $\gamma$ to simply notation. Fix $\varepsilon>0$, and separate the integral in question as:
 \begin{flalign}
  &\int \int \mathds{1}\{\|\eta-\eta_0\|>\varepsilon\}p_{\smi,\gamma}(\eta,\psi\mid\mathcal{D})\dt\psi\dt\eta\nonumber\\&=  
  \int \int_{\|\psi-\psi_0\|>C M^\alpha_n r_n^\alpha} \mathds{1}\{\|\eta-\eta_0\|>\varepsilon\}p_{\smi,\gamma}(\eta,\psi\mid\mathcal{D})\dt\psi\dt\eta\nonumber\\&+\int \int_{\|\psi-\psi_0\|\le C M^\alpha_n r_n^\alpha} \mathds{1}\{\|\eta-\eta_0\|>\varepsilon\}p_{\smi,\gamma}(\eta,\psi\mid\mathcal{D})\dt\psi\dt\eta\label{eq:two_terms}.
 \end{flalign}
The first term satisfies the following trivial bound: for $n$ and $M$ large enough,
\begin{flalign}
&\int \int_{\|\psi-\psi_0\|>C M^\alpha_n r_n^\alpha} \mathds{1}\{\|\eta-\eta_0\|>\varepsilon\}p_{\smi,\gamma}(\eta,\psi\mid\mathcal{D})\dt\psi\dt\eta\nonumber\\&\le \int_{\|\psi-\psi_0\|>C M^\alpha_n r_n^\alpha}p_{\smi,\gamma}(\psi\mid\mathcal{D})\dt\psi\nonumber\\&\le 1/M,\label{eq:term1}
\end{flalign}
with probability converging to one by Theorem \ref{thm:phi_conc}.

For the second term, we have that 
\begin{flalign*}
&\int \int_{\|\psi-\psi_0\|\le C M^\alpha_n r_n^\alpha} \mathds{1}\{\|\eta-\eta_0\|>\varepsilon\}p_{\smi,\gamma}(\eta,\psi\mid\mathcal{D})\dt\psi\dt\eta \\&\le \sup_{\|\psi-\psi_0\|\le C M^\alpha_n r_n^\alpha}\int \mathds{1}\{\|\eta-\eta_0\|>\varepsilon\}p_{\smi,\gamma}(\eta\mid\psi,\mathcal{D})\dt\eta.
\end{flalign*}Recalling the definition  of $p_{\smi,\gamma}(\eta\mid\psi,\mathcal{D})$ we have
$$
\int \mathds{1}\{\|\eta-\eta_0\|>\varepsilon\} p_{\smi,\gamma}(\eta\mid\psi,\mathcal{D})\dt\eta=\frac{\int \mathds{1}\{\|\eta-\eta_0\|>\varepsilon\}p(\eta\mid\psi)e^{-n\ell_n(\eta,\psi)}\dt\eta}{\int p(\eta\mid\psi)e^{-n\ell_n(\eta,\psi)}\dt\eta}=:\frac{N_\psi(\varepsilon)}{D_\psi}
$$

We now deal with the numerator and denominator separately. 

\noindent\textbf{Denominator.} Consider an  $h\in\mathbb{R}^{d_\eta}$ such that $\|h\|\le n^{-1/2}$. Then, for $n$ large enough, by Assumption 5(i), 
$
\eta_\star(\psi)+h\in\mathcal{E}
$. From the definition of $D_\psi$, we have
\begin{flalign*}
D_\psi=\int_{\mathcal{E}}p(\eta\mid\psi)e^{-n\ell_n(\eta,\psi)} \dt\eta\ge \int_{\|\eta-\eta_\star(\psi)\|\le\delta}p(\eta\mid\psi)e^{-n\ell_n(\eta,\psi)}   \dt\eta\ge \underline{c}_p\int_{\|\eta-\eta_\star(\psi)\|\le\delta}e^{-n\ell_n(\eta,\psi)}   \dt\eta
\end{flalign*}where the second inequality follows from Assumption 5(iii). For $
R_n:=\sup_{\eta,\psi\in\mathcal{E}\times\Phi}|\ell_n(\eta,\psi)-\mathcal{L}(\eta,\psi)|
$, we have 
\begin{flalign*}
\ell_n(\eta,\psi)&\le R_n+\mathcal{L}_\infty(\eta,\psi)
\end{flalign*}and, from Assumption 5(ii), on the set $\|\eta-\eta_\star(\psi)\|\le\delta$, 
$$
\mathcal{L}(\eta,\psi)\le \mathcal{L}\{\eta_\star(\psi),\psi\}+\overline{c}\|\eta-\eta_\star(\psi)\|^2, 
$$so that we have 
\begin{flalign*}
\ell_n(\eta,\psi)&\le R_n+\mathcal{L}_\infty(\eta,\psi)\\&= R_n+\mathcal{L}_\infty(\eta,\psi)-\mathcal{L}_\infty(\eta_\star(\psi),\psi)+\mathcal{L}_\infty(\eta_\star(\psi),\psi)\\&\le R_n+\overline{c}\|\eta-\eta_\star(\psi)\|^2+\mathcal{L}_\infty(\eta_\star(\psi),\psi).
\end{flalign*}For $n$ large enough, $\{h:\eta_\star(\psi)+h\}\subseteq\{\eta:\|\eta-\eta_\star(\psi)\|\le\delta\}$, and so we have that 
\begin{flalign*}
\ell_n(\eta_\star(\psi)+h,\psi)&\le R_n+\overline{c}\|h\|^2+\mathcal{L}_\infty(\eta_\star(\psi),\psi).
\end{flalign*}Applying this to the denominator yields
\begin{flalign*}
D_\psi\ge \underline{c}_p e^{-n\mathcal{L}_\infty\{\eta_\star(\psi),\psi)\}-nR_n}\int_{\|h\|\le n^{-1/2}}e^{-\overline{c}\|h\|^2}\dt h\ge C_1 n^{-d_\eta/2}e^{-n\mathcal{L}_\infty\{\eta_\star(\psi),\psi)\}-nR_n}
\end{flalign*}for some $C_1>0$. 

\noindent\textbf{Numerator.} For the numerator, we first use Assumption 5(iii) to write 
$$
N_\psi(\varepsilon)\le \overline{c}_p\int\mathds{1}\{\|\eta-\eta_\star(\psi)\|>\varepsilon\} e^{-n\ell_n(\eta,\psi)}\dt\eta.
$$For $R_n$ as defined earlier, we can apply Assumption 5(ii) to obtain
$$
-\ell_n(\eta,\psi)\le R_n-\mathcal{L}_\infty\{\eta_\star(\psi),\psi)\}-\underline{c}\|\eta-\eta_\star(\psi)\|^2,
$$so that 
\begin{flalign*}
N_\psi(\varepsilon)\le \overline{c_p}e^{n\cdot R_n-\mathcal{L}_\infty\{\eta_\star(\psi),\psi\}}\int_{\|h\|>\varepsilon}e^{-\underline{c}\|h\|^2}\dt h\le C_2 n^{-d_\eta/2}e^{n\cdot R_n-\mathcal{L}_\infty\{\eta_\star(\psi),\psi\}}e^{-\frac{\underline{c}}{2}n\varepsilon^2}
\end{flalign*}for some $C_2>0$.

\noindent\textbf{Combining Terms.} Combining terms we have 
\begin{flalign*}
    \int \mathds{1}\{\|\eta-\eta_0\|>\varepsilon\} p_{\smi,\gamma}(\eta\mid\psi,\mathcal{D})=\frac{N_\psi(\varepsilon)}{D_\psi}\le C\exp\left\{2 n\cdot R_n-\frac{\underline{c}}{2}n\varepsilon^2\right\}.
\end{flalign*}Define the set $V_{n,M}:=\{R_n\le  M_n r_n\}$ for some $M_n>0$, possibly $M_n\rightarrow\infty$. Set $\varepsilon=M_n\sqrt{r_n+n^{-1}}$; on the set $V_{n,M}$, we then have, 
$$
2n R_n-\frac{\underline{c}}{2}n\varepsilon^2\le 2M nr_n-\frac{\underline{c}}{2}n\varepsilon^2=2M_nnr_n-\frac{\underline{c}}{2}M_n^2n(r_n+n^{-1}),
$$ which, for $M_n\rightarrow\infty$, diverges as $n\rightarrow\infty$. Hence, on $V_{n,M}$ we have that 
$$
\sup_{\|\psi-\psi_0\|\le CM_n^\alpha r_n^\alpha}\int \mathds{1}\{\|\eta-\eta_0\|>\varepsilon\}p_{\smi,\gamma}(\eta\mid\psi,\mathcal{D})\dt\eta\rightarrow0. 
$$Note that, from Assumption 3 and Markov's inequality, $\Pr(V^c_{n,M})\le C/M_n$, for some $C>0$. Hence, for any $\epsilon>0$, 
\begin{flalign*}
&\limsup_{n\rightarrow\infty}\Pr\left\{\sup_{\|\psi-\psi_0\|\le CM_n^\alpha r_n^\alpha}\int \mathds{1}\{\|\eta-\eta_0\|>\varepsilon\}p_{\smi,\gamma}(\eta\mid\psi,\mathcal{D})\dt\eta>\epsilon\right\}\\&\le 
\limsup_{n\rightarrow\infty}\Pr\left\{V_{n,M}^c\right\}\\&\le C/M_n.
\end{flalign*}
\end{proof}

\subsection{Primitive conditions for Assumption \ref{ass:entropy}}\label{sec:primitive}
Assumption 3 is a high-level condition that is somewhat hard to interpret. Herein, we show that Assumption 3 is satisfied by lower-level conditions on the marginal models and copula function which can be checked on a per-model basis.

\begin{assumption}\label{ass:cop_cont}
    The copula model density is continuous in $u$ for all $u\in(0,1)$.
\end{assumption}
\begin{assumption}\label{ass:suff_psi}The set $\Psi$ is a bounded subset of $\mathbb{R}^{d_\psi}$.
    For any $w\in(0,1)^d$, and any $\psi,\psi'\in\Psi$, $|\log c(w;\psi)-\log c(w;\psi')|\le L(w)\|\psi-\psi'\|$, with $\E[L(W)^2]<\infty$. Lastly, $\E[\sup_{\psi\in\Psi}|\log c(W;\psi)|^2]<\infty$. 
\end{assumption}
\begin{assumption}\label{ass:suff_eta}The set $\mathcal{E}$ is a bounded subset of $\mathbb{R}^{d_\eta}$.
    For any $y\in\mathcal{Y}$, and any $\eta,\eta'\in\mathcal{E}$, $|\log f(y;\eta)-\log f(y;\eta')|\le L(y)\|\eta-\eta'\|$, with $\E[L(Y)^2]<\infty$. Lastly, $\E[\sup_{\eta\in\mathcal{E}}|\log f(Y;\eta)|^2]<\infty$. 
\end{assumption}
\begin{lemma}\label{lem:suff_copula}
Under Assumptions \ref{ass:cop_cont}-\ref{ass:suff_eta},  Assumption \ref{ass:entropy} is satisfied with $r_n=\log(n)/\sqrt{n}$.  
\end{lemma}
\begin{proof}[Proof of Lemma \ref{lem:suff_copula}]
Consider the SMI pseudo-likelihood
$$
\Delta_n(\eta,\psi;\gamma):=\prod_{i=1}^n \Delta_{\tilde{a}_{i 1}}^{\tilde{b}_{i 1}} \cdots \Delta_{\tilde{a}_{i d}}^{\tilde{b}_{i d}} C(u_i ; \psi),
$$
where $\tilde{a}_{i j}=\gamma_j u_{i j}+\left(1-\gamma_j\right) a_{i j}, \tilde{b}_{i j}=\gamma_j u_{i j}+\left(1-\gamma_j\right) b_{i j}, u_{i j}=F_j\left(y_{i j} ; \eta_j\right), a_{i j}=R\left(y_{i j}\right) /(n+1)$ and $b_{i j}=\left(R\left(y_{i j}\right)+1\right) /(n+1)$.

Now, the $i$-th term in $\Delta_n(\eta,\psi;\gamma)$ can be represented as
$$
\int_{\tilde{a}_{i 1}}^{\tilde{b}_{i 1}}\cdots \int_{\tilde{a}_{i d}}^{\tilde{b}_{i d}} c(u ; \psi) d u
$$
where $c(u;\psi)$ is the copula model density. Since $c(u ; \psi)$ is continuous in $u$, from the mean-value theorem, there exists 
$$
w_i=\left(w_{i 1}, \dots,w_{i d}\right) \in\bigtimes_{j=1}^d\left[\tilde{a}_{i j}, \tilde{b}_{i j}\right]
$$ such that
$$
\frac{1}{\prod_{j=1}^{d}\left(\tilde{b}_{i j}-\tilde{a}_{i j}\right)} \int_{\tilde{a}_{i 1}}^{\tilde{b}_{i 1}}\cdots \int_{\tilde{a}_{i d}}^{\tilde{b}_{i d}} c(u ; \psi) \dt u=c\left(w_i ; \psi\right)
$$
Noting that
$$
\prod_{j=1}^{d}\left(\tilde{b}_{i j}-\tilde{a}_{i j}\right)=\prod_{j=1}^d \frac{1-\gamma_j}{n+1},
$$
we have that
$$
\Delta_n(\eta,\psi;\gamma)=\left\{\prod_{i=1}^n \prod_{j=1}^d \frac{1-\gamma_j}{n+1}\right\}\prod_{i=1}^n c\left(w_i ; \psi\right).
$$
Since the posterior $p_{\mathrm{SMI},\gamma}(\psi,\eta\mid\mathcal{D})$ is as a ratio of the pseudo-likelihood $\Delta_n(\eta,\psi;\gamma)$, the terms  $\left\{\prod_{i=1}^n \prod_{j=1}^d \frac{1-\gamma_j}{n+1}\right\}$, and since loss functions are arbitrary up to an additive scale - and since the loss is defined with logarithms - we abuse notation in what follows and set  
$$
\Delta_n(\eta,\psi;\gamma):=\prod_{i=1}^n c\left(w_i ; \psi\right).
$$ 
Using this revised definition, Assumption \ref{ass:entropy} can be re-stated and bounded as
\begin{flalign*}
&\E\sup_{\eta,\psi}|\MD_n(\varphi)-\MD(\varphi)|\\&=
\E\sup_{\eta,\psi}|-n^{-1}\{\log\Delta_n^\gamma(\eta,\psi)-\E\log\Delta_n^\gamma(\eta,\psi)\}-n^{-1}\{\log f_n(\eta)-\E\log f_n(\eta)\}|
\\&\le \E\sup_{\eta,\psi}|n^{-1}\{\log\Delta_n^\gamma(\eta,\psi)-\E\log\Delta_n^\gamma(\eta,\psi)\}|+\E\sup_{\eta,\psi}|n^{-1}\{\log f_n(\eta)-\E\log f_n(\eta)\}|
\end{flalign*}
From Theorem 2.7.11 in \cite{vaart2023empirical}, Assumptions \ref{ass:suff_psi}-\ref{ass:suff_eta} are sufficient to ensure that 
$$
\E\sup_{\eta,\psi}|\MD_n(\varphi)-\MD(\varphi)|\lesssim \log(n)/\sqrt{n}.
$$

\end{proof}

\section{Additional simulations}
To illustrate the scalability of the SMI approach we consider an additional simulation where the degree of misspecification varies across marginals. In contrast to the bivariate example discussed in Section~5, here the influence parameter $\gamma=(\gamma_1,\dots,\gamma_d)^\top$ is $d=15$ dimensional. Additionally, this example illustrates how $\gamma$ can be selected through a task-specific utility function when the quantity of interest is a functional of all modules jointly.

\textbf{Data generating process:} 
We generate $n=1,000$ observations from a $d=15$ dimensional Gaussian copula model. The correlation matrix parameterizing the Gaussian copula is specified through a one-factor model, 
$$
\Omega=LL^\top+D,
$$
where $L=(l_1,\dots,l_d)^\top$ is a vector of factor loadings with $-1\leq l_j\leq 1$ and $D$ is a diagonal matrix with $j$-th entry $1-l_j^2$, $j=1,\dots,d$. The $(i,j)$-th off-diagonal entry of $\Omega$ is thus $\Omega_{ij}=l_i l_j$, and the diagonal entries equal one, so that $\Omega$ is a valid correlation matrix. To ensure identifiability of the parameters we restrict $l_1>0$. 
The marginals of the true data generating process are divided into five blocks, each block containing a well-specified, a mildly misspecified, and a fully misspecified marginal. For $j=1,\dots,5$, they are generated as 
$$
Y_{3j-2}\sim\ND(0,1),\quad Y_{3j-1}\sim\skewND(0,1,-5),\quad Y_{3j}\sim 0.6\,\ND(-1,0.4^2)+0.4\,\ND(0,0.2^2),
$$
where $\skewND(\xi,\omega,\alpha)$ is a skew-normal distribution with location $\xi$, scale $\omega$, and skewness parameter $\alpha$. Matching this block structure, the factor loadings are $l_{3j-2}=l_{3j-1}=l_{3j}=0.65+0.05\,j$ for $j=1,\dots,5$, so that within-block dependence increases across blocks.

\textbf{Statistical model:} 
We fit a joint Gaussian distribution, i.e. a Gaussian copula with Gaussian marginals. The copula is therefore correctly specified, while the marginals are correctly specified only for the components $Y_{3j-2}$; the components $Y_{3j-1}$ are mildly misspecified through the skewness, and the components $Y_{3j}$ are severely misspecified through the bimodal mixture. The model parameters are the factor loadings $l_1,\dots,l_d$ of the one-factor correlation structure, together with the marginal means and variances $(\mu_j,\sigma_j^2)$, $j=1,\dots,d$. Setting $l_1=\tanh(\exp(\psi_1))$ and $l_j=\tanh(\psi_j)$ for $j=2,\dots,d$ gives an unrestricted vector of copula parameters. Setting further $\eta_j=(\mu_j,\log\sigma_j)^\top$, $\eta=(\eta_1^\top,\dots,\eta_d^\top)^\top$, gives the joint vector of unconstrained model parameters $\theta=(\psi^\top,\eta^\top)^\top$. The priors $\tanh^{-1}(l_1)\sim \HND(0,2^2)$, $\tanh^{-1}(l_j)\sim\ND(0,2^2)$ for $j>1$, $\mu_j\sim\ND(0,2.5^2)$, and $\sigma_j\sim\HND(0,10^2)$, where $\HND$ is a half-normal distribution, complete the model.

\textbf{Selecting $\gamma$:} 
To illustrate the use of an alternative utility function, we assume that interest lies in the aggregate 
$$
X=w_1 Y_1+\dots+w_d Y_d,\qquad \sum_{j=1}^d w_j=1,
$$
with weights derived from the fitted statistical model. Weighted summations of variables are common objectives for optimization in a variety of problems, including in financial risk management~\citep{Markowitz1952} and advertising reach~\citep{Danaher1991}. Under a joint Gaussian distribution, the weights minimizing the variance of $X$ depend only on the correlation structure induced by $L$ and on the marginal variances, and are given in closed form by
$$
w(\theta)=\frac{\Sigma(\theta)^{-1}\mathbf{1}}{\mathbf{1}^\top\Sigma(\theta)^{-1}\mathbf{1}},
\qquad
\Sigma(\theta)=S(\eta)\,\Omega(\psi)\,S(\eta),
$$
where $S(\eta)=\operatorname{diag}(\sigma_1,\dots,\sigma_d)$ and $\mathbf 1$ is the vector of ones~\citep{Merton1972}. The weights $w(\theta)=(w_1(\theta),\dots,w_d(\theta))^\top$ are thus a function of the model parameters. The weights depend critically on the off-diagonal elements of $\Sigma$, which are directly informed by the factor loadings. This is exactly the area where the Bias-Bias trade off induced by the SMP is most influential. A natural choice is thus to select the influence parameters $\gamma$ so that the empirical variance of the observed weighted summations $x_i(\theta)=w(\theta)^\top y_i$ is minimized. This results in the the utility function 
\begin{equation*}
    u(\gamma)=-\int \frac{1}{n-1}\sum_{i=1}^n (w(\theta)^\top y_i - \overline{x(\theta)})^2
    \,
    p_{\text{SMI},\gamma}(\theta\mid\mathcal{D})\,d\theta,
\end{equation*}
where $\overline{x(\theta)}=n^{-1}\sum_{i=1}^n x_i(\theta)$. As in Section~6, we consider 100 steps of BO to find the optimal SMI parameters $\gamma^*=\arg\max_\gamma u(\gamma)$. 

\textbf{Results:} 
Figure~\ref{fig:sim2_utility} shows how the utility improves with the number of SMPs checked over the $100$ BO steps. The optimal SMP results in a higher utility (i.e. with weights giving a smaller variance) than both the conventional and the fully cut posterior. A true SMP with most influence parameters in $(0,1)$ is selected, further illustrating that the SMI framework is beneficial in models with varying degrees of misspecification. The three posteriors result in substantially different weight distributions, which is mostly due to different estimates of the correlation matrix $\Omega$. Figure~\ref{fig:sim2_loadings} shows the marginal posteriors for the loadings $L=(l_1,\dots,l_d)^\top$. Differences between the three posteriors are more pronounced for higher degrees of dependence (left to right) and higher degrees of marginal misspecification (top to bottom). As described in Section~4, different SMPs concentrate on different points resulting in appreciably different point estimates. $\Omega(\widehat{\psi})$ is closest to the true correlation matrix $\Omega^*$ from the data generating process under the SMP with $||\Omega(\widehat{\psi})-\Omega^*||_F=0.1470$ (conventional: 0.3353, cut: 0.1707). While there are strong differences in the estimation of the copula parameters, differences in the estimation of the marginals are less pronounced as illustrated by Figure~\ref{fig:sim2_marginals}.

\begin{figure}[tbh!]
    \centering
    \includegraphics[width=0.95\linewidth,keepaspectratio]{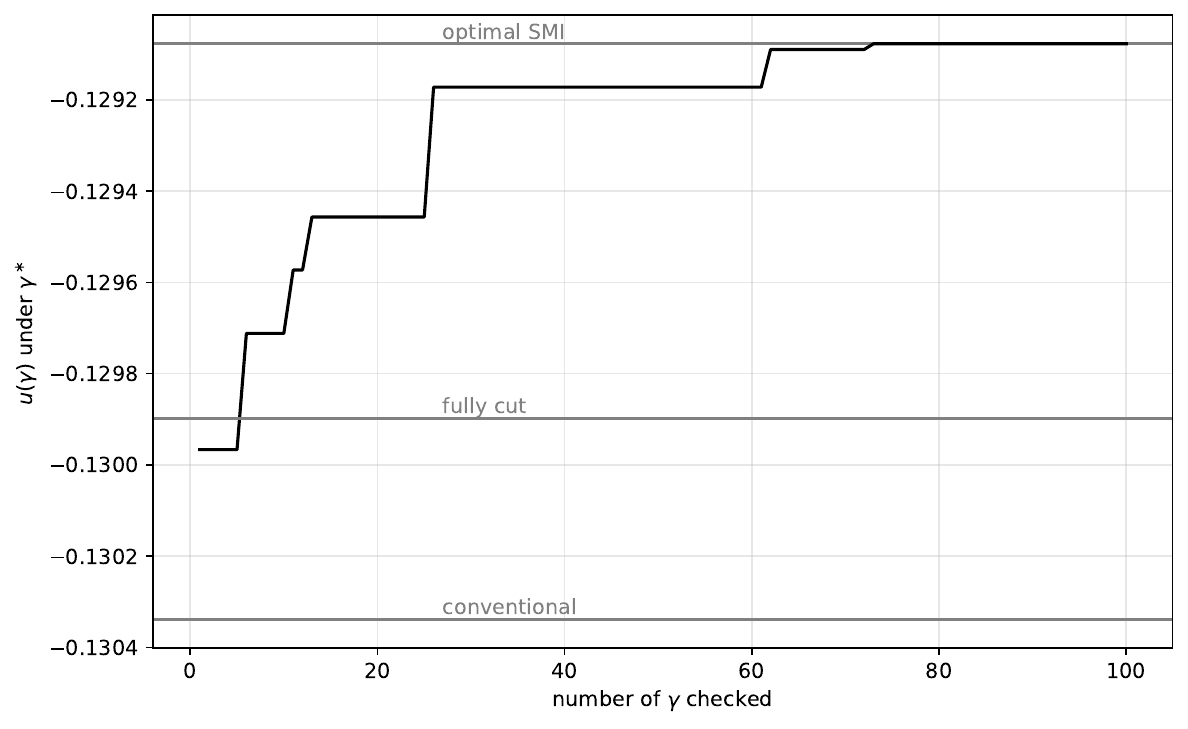}
    \caption{\small Convergence of BO over $100$ iterations for the additional simultaion study. The vertical axis is the utility $u(\gamma)$, and the black line indicates the optimal value $u(\gamma^\ast)$ against number of $\gamma$ values explored. The values under the conventional and the fully cut posterior are plotted as gray horizontal lines. 
    }
    \label{fig:sim2_utility} 
\end{figure}

\begin{figure}[tbh!]
    \centering
    \includegraphics[width=0.95\linewidth,keepaspectratio]{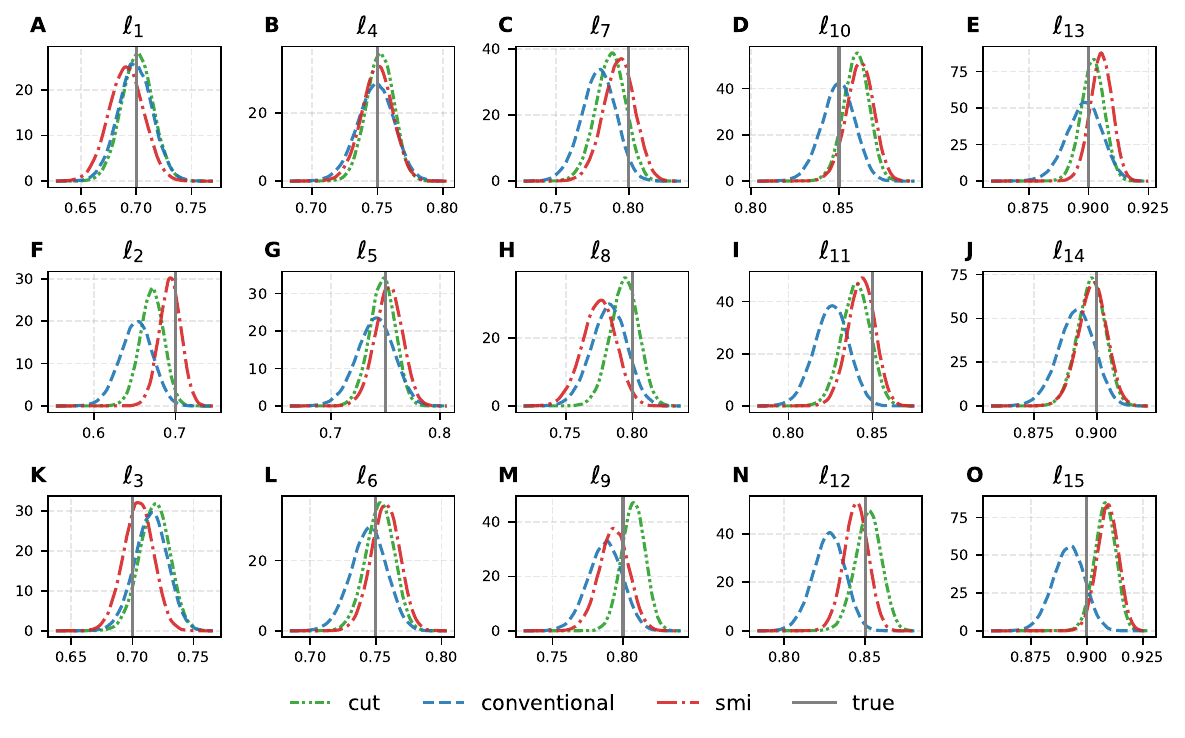}
    \caption{\small Plots of the marginal posterior distributions for the factor loadings. Densities are estimates using a large sample from the variational posterior under the conventional posterior (blue dashed line), the fully cut posterior (green dotted line), and the optimal copula-SMP (red dashed-dot line). True loadings from the data generating process are given by horizontal lines. 
    }
    \label{fig:sim2_loadings} 
\end{figure}

\begin{figure}[tbh!]
    \centering
    \includegraphics[width=0.95\linewidth,keepaspectratio]{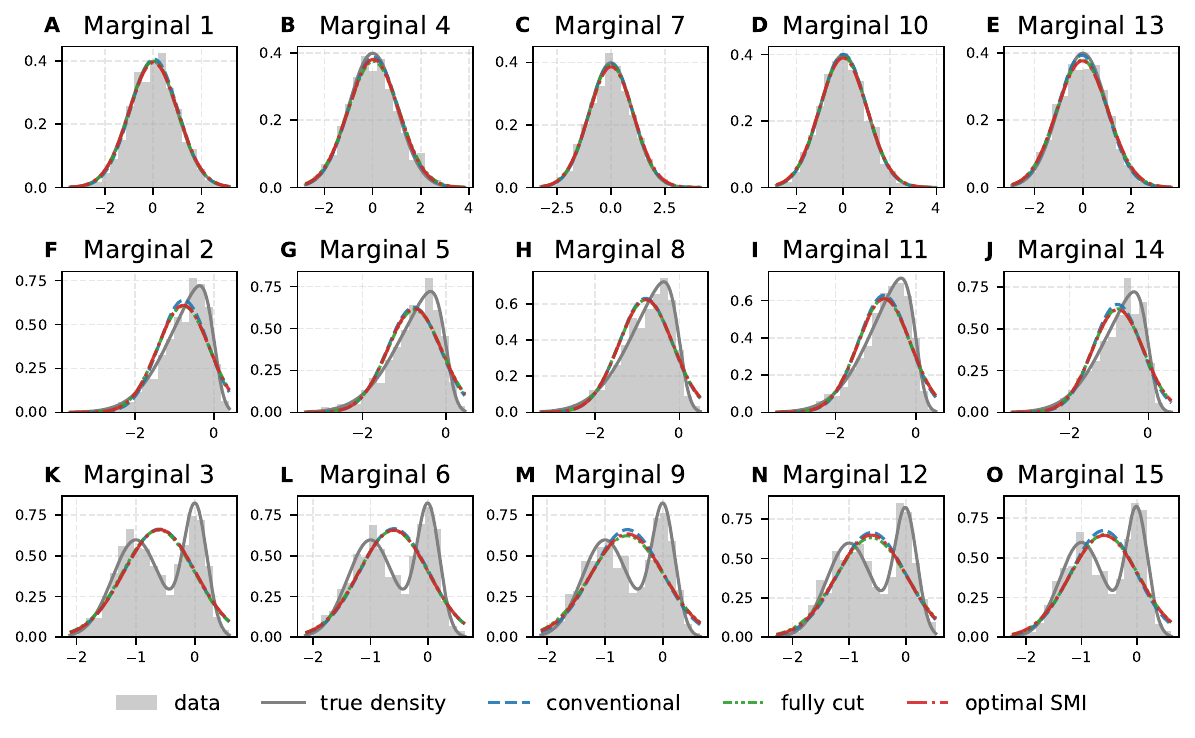}
    \caption{\small Plots of the estimated marginal densities in the additional simulation study. Histograms of the data and true marginal densities are given in gray, and the density estimates are evaluated at the posterior means $\widehat{\eta}$ for the conventional posterior (blue dashed line), the fully cut posterior (green dotted line), and the optimal copula-SMP (red dashed-dot line).
    }
    \label{fig:sim2_marginals} 
\end{figure}

\section{Additional information on the bond yields application}

\subsection{Parameterization of the skew-normal copula} \label{app:var_copula}
Assume $\widetilde Y_t=(Y_{\text{VIX},t},Y_{\text{AAA},t},Y_{\text{BBB},t})^\top$ follows a  stationary stochastic process. Then from Theorem~1 in~\cite{Smi2015}, the three
marginals of $\widetilde{Y}_t$ are time invariant, and the $d=6$ dimensional copula of
$Y_t=(\widetilde Y_t^\top,\widetilde Y_{t-1}^\top)^\top$ is also stationary, in that $C(u_t,u_{t-1})$ is invariant to $t$. Invariance in the implicit SN copula $C$ follows from invariance in the underlying SN distribution, which can be enforced by the two parameter constraints described below.

The first is constraining  
the $(d\times d)$ scale matrix of the SN to be block Toeplitz, so that
\begin{equation}
\Omega=\begin{pmatrix}
        \Omega_0 & \Omega_1\\
        \Omega_1^\top & \Omega_0 
    \end{pmatrix}\,.\label{eq:toeplitz}
\end{equation}
An efficient way to impose this is to extend
the strategy employed by~\citet[Sec.3]{SmiVah2016} for Gaussian copulas as follows. Because $\Omega$ is a correlation matrix (which is necessary to identify the copula parameters) there is a one-to-one relationship between the entries of $\Omega$ and the (semi) partial correlations $\{\rho_{i,j};j=1,\ldots,d;i<j\}$ that are due to Yule and outlined in~\cite{Joe2006} and \cite{DanPou2009}. The partials $\rho_{i,j}$ are unconstrained on $[-1,1]$ and the Toeplitz structure is recovered by setting $b=d/2$ and $\rho_{i,j}=\rho_{b+i,b+j}$ for $1\leq i < j \leq b$ with $b=3$ in our example. This results in $b(b-1)/2$ unique partials to parameterize $\Omega_0$ and $b^2$ unique partials to parameterize $\Omega_1$. To compute inference we transform each partial to the real line as $\tilde{\rho}_{i,j}=\Phi^{-1}((\rho_{i,j}+1)/2)$.  

The second constraint follows from writing $\alpha = \left(1-\delta^\top\Omega^{-1}\delta\right)^{-\frac{1}{2}}\Omega^{-1}\delta$ in terms of the slant parameters $\delta$ in~\cite{AzzVal1996}. If we partition $\delta=(\delta_2^\top,\delta_1^\top)^\top$ and $\alpha=(\alpha_2^\top,\alpha_1^\top)^\top$ consistent with~\eqref{eq:toeplitz}, then 
invariance follows when $\delta_2=\delta_1$. This implies that $\alpha_{1}=(\Omega_0-\Omega_1)^{-1}(\Omega_0-\Omega_1^\top)\alpha_2$. All together,  $\Omega_0$ is parameterized through 3 unique transformed partials, $\Omega_1$ is parameterized through 9 unique transformed partials, and $\alpha$ is parameterized by 3 unique values in $\alpha_2$, giving an unrestricted copula parameter vector $\psi$ of dimension 15. Because there are only three unique marginals, $\gamma$ is of dimension three.

\subsection{Additional results}\label{app:application}
This section presents additional results for the financial application discussed in Section~6 of the manuscript. 

Figure~\ref{fig:app_opt} shows how the utility improves with the number of SMPs checked. In this example, BO converges fast after a brief exploration of the $\gamma$ space.

Table~\ref{tab:app} reports $\Delta_\text{Major}(0.1)$ and $\Delta_\text{Minor}(0.1)$, as well as Kendall's $\tau$ and the parameters of the marginal skew-Gaussian copula for all pairs evaluated at the posterior means, $\widehat{\psi}$, for the conventional posterior, the fully cut posterior, and the optimal SMP. $\Delta_\text{Major}(0.1)$, $\Delta_\text{Minor}(0.1)$, and Kendall's $\tau$ are estimated using a large sample from the copula model.

Figures~\ref{fig:app_dependence_smi}, \ref{fig:app_dependence_conventional}, and \ref{fig:app_dependence_cut} give quantile dependence plots for all pairs for the SMI, the conventional, and the fully cut posterior respectively.

\begin{figure}[tbh!]
    \centering
    \includegraphics[width=0.95\linewidth,keepaspectratio]{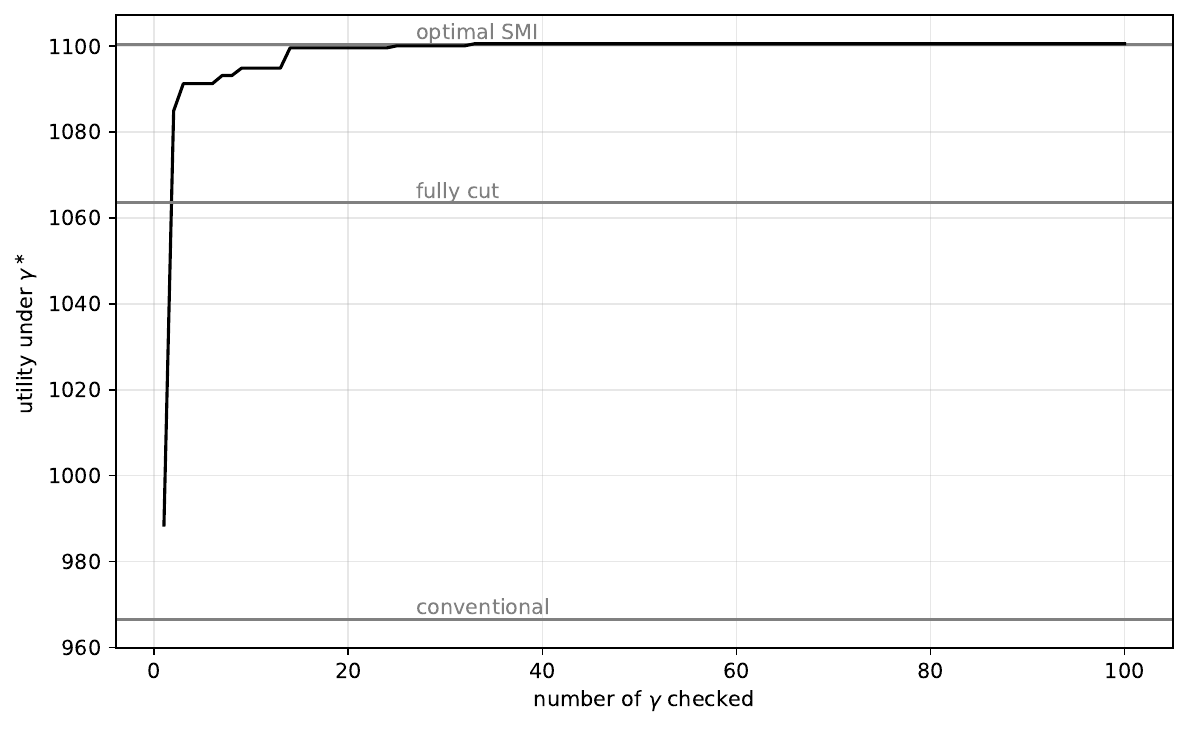}
    \caption{\small Convergence of BO over $100$ iterations for the bond yields example. The vertical axis is the utility $u(\gamma)$, and the black line indicates the optimal value $u(\gamma^\ast)$ against number of $\gamma$ values explored. The values $u((1,1,1)^\top)$, $u((0,0,0)^\top)$ are plotted as gray horizontal lines. 
    }
    \label{fig:app_opt} 
\end{figure}

\begin{table}[ph!]
\centering
\begin{adjustbox}{max width=\textwidth, max height=0.43\textheight}
{
\renewcommand{\arraystretch}{0.7}
\begin{tabular}{ccccccc}&  &  & \multicolumn{2}{c}{10 \% Quantile asymmetry} & \multicolumn{2}{c}{Parameters} \\& Pair & Kendall's $\tau$ & $\Delta_\text{major}(0.1)$ & $\Delta_\text{minor}(0.1)$ & $\rho$ & $\alpha$ \\ \hline
conventional & \multirow{3}{*}{VIX--AAA} &  -0.2957 & -0.0003 & 0.0059 & -0.4486 & $( 0.0014 , -0.0115 )^\top$ \\
cut & &  -0.2577 & 0.0158 & -0.1542 & -0.6185 & $( 0.2725 , -4.7911 )^\top$ \\
optimal SMP & &  -0.2487 & 0.0215 & -0.2082 & -0.643 & $( 0.5786 , -5.1495 )^\top$ \\
\hline
conventional & \multirow{3}{*}{VIX--BBB} &  -0.2104 & 0.0008 & 0.002 & -0.3243 & $( 0.0039 , -0.0082 )^\top$ \\
cut & &  -0.1812 & 0.0165 & -0.1266 & -0.534 & $( 0.5458 , -1.6735 )^\top$ \\
optimal SMP & &  -0.1725 & 0.015 & -0.1493 & -0.5527 & $( 0.7125 , -1.6371 )^\top$ \\
\hline
conventional & \multirow{3}{*}{VIX--VIXlag} &  0.8355 & -0.001 & 0.0 & 0.9669 & $( 0.0033 , 0.0033 )^\top$ \\
cut & &  0.8115 & 0.0217 & 0.0 & 0.968 & $( 0.4245 , 0.4245 )^\top$ \\
optimal SMP & &  0.833 & 0.0206 & 0.0 & 0.9763 & $( 0.4865 , 0.4865 )^\top$ \\
\hline
conventional & \multirow{3}{*}{VIX--AAAlag} &  -0.3003 & -0.0002 & 0.0035 & -0.4552 & $( 0.0013 , -0.0115 )^\top$ \\
cut & &  -0.2653 & 0.0141 & -0.1515 & -0.624 & $( 0.2293 , -4.7839 )^\top$ \\
optimal SMP & &  -0.2548 & 0.0196 & -0.2016 & -0.6471 & $( 0.5384 , -5.1084 )^\top$ \\
\hline
conventional & \multirow{3}{*}{VIX--BBBlag} &  -0.2172 & 0.0004 & 0.0019 & -0.3345 & $( 0.0038 , -0.0082 )^\top$ \\
cut & &  -0.1915 & 0.0142 & -0.1224 & -0.5438 & $( 0.522 , -1.664 )^\top$ \\
optimal SMP & &  -0.183 & 0.0145 & -0.147 & -0.5618 & $( 0.6889 , -1.6229 )^\top$ \\
\hline
conventional & \multirow{3}{*}{AAA--BBB} &  0.7745 & -0.0035 & 0.0 & 0.9376 & $( -0.0264 , 0.0153 )^\top$ \\
cut & &  0.6917 & -0.0816 & 0.0 & 0.9408 & $( -9.8266 , 3.0416 )^\top$ \\
optimal SMP & &  0.7013 & -0.0514 & 0.0 & 0.9422 & $( -15.0267 , 5.1941 )^\top$ \\
\hline
conventional & \multirow{3}{*}{AAA--VIXlag} &  -0.2932 & -0.0009 & -0.0059 & -0.4449 & $( -0.0115 , 0.0014 )^\top$ \\
cut & &  -0.2531 & 0.0164 & 0.158 & -0.6155 & $( -4.7979 , 0.2964 )^\top$ \\
optimal SMP & &  -0.2455 & 0.0221 & 0.2125 & -0.6408 & $( -5.1739 , 0.6005 )^\top$ \\
\hline
conventional & \multirow{3}{*}{AAA--AAAlag} &  0.9375 & 0.0004 & 0.0 & 0.9952 & $( -0.0061 , -0.0061 )^\top$ \\
cut & &  0.9063 & -0.1052 & 0.0 & 0.9964 & $( -2.4817 , -2.4818 )^\top$ \\
optimal SMP & &  0.9167 & -0.0909 & 0.0 & 0.9972 & $( -2.5738 , -2.5736 )^\top$ \\
\hline
conventional & \multirow{3}{*}{AAA--BBBlag} &  0.7689 & -0.0037 & 0.0 & 0.9345 & $( -0.0255 , 0.0143 )^\top$ \\
cut & &  0.6819 & -0.09 & 0.0 & 0.938 & $( -8.8588 , 2.547 )^\top$ \\
optimal SMP & &  0.6917 & -0.063 & 0.0 & 0.9394 & $( -12.2755 , 3.9788 )^\top$ \\
\hline
conventional & \multirow{3}{*}{BBB--VIXlag} &  -0.2059 & -0.0005 & -0.0025 & -0.3173 & $( -0.0083 , 0.0039 )^\top$ \\
cut & &  -0.1736 & 0.0175 & 0.1296 & -0.5272 & $( -1.6809 , 0.5623 )^\top$ \\
optimal SMP & &  -0.1652 & 0.0161 & 0.152 & -0.5462 & $( -1.648 , 0.7296 )^\top$ \\
\hline
conventional & \multirow{3}{*}{BBB--AAAlag} &  0.7639 & -0.0067 & 0.0 & 0.9316 & $( 0.0135 , -0.0246 )^\top$ \\
cut & &  0.6761 & -0.0954 & 0.0 & 0.9361 & $( 2.2893 , -8.3689 )^\top$ \\
optimal SMP & &  0.6871 & -0.0593 & 0.0 & 0.9381 & $( 3.5983 , -11.4407 )^\top$ \\
\hline
conventional & \multirow{3}{*}{BBB--BBBlag} &  0.9402 & -0.005 & 0.0 & 0.9956 & $( -0.0048 , -0.0048 )^\top$ \\
cut & &  0.9097 & -0.036 & 0.0 & 0.995 & $( -0.8979 , -0.8979 )^\top$ \\
optimal SMP & &  0.9095 & -0.0307 & 0.0 & 0.9949 & $( -0.8787 , -0.8787 )^\top$ \\
\hline
\end{tabular}
}
\end{adjustbox}
\caption{\small Dependence metrics for variable pairs in the bond yields example. Kendall's $\tau$, asymmetry in the major and minor diagonals at level $\zeta=0.1$, as well as the parameters of the marginal skew-normal copula (all rounded to 4 digits), all evaluated at the posterior mean $\widehat{\psi}$. Results are given for variable pairs and for the conventional posterior (i.e. $\gamma=(1,1,1)$, the fully cut posterior (i.e. $\gamma=(0,0,0)^\top$)), and the optimal SMP.}
\label{tab:app}
\end{table}

\begin{figure}[tbh]
    \centering
    \includegraphics[width=0.95\linewidth,keepaspectratio]{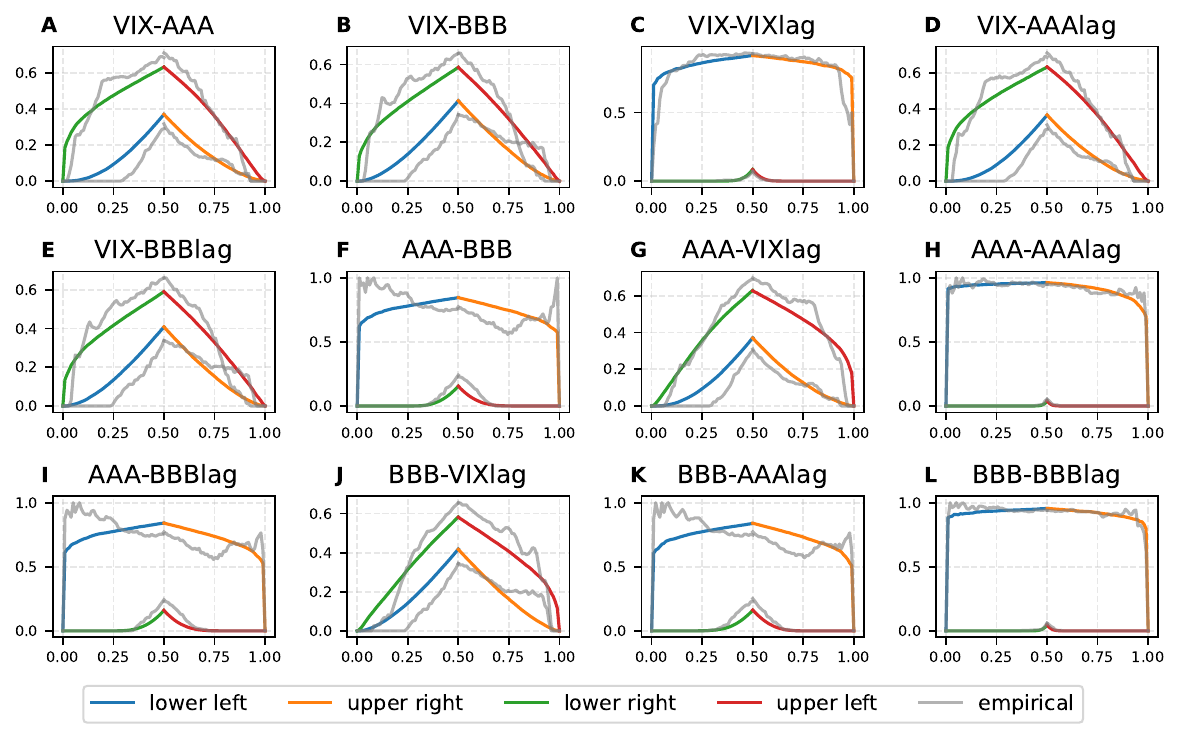}
    \caption{\small Quantile plots from the SMP in the bold yields example. The quantiles are $\rho_{LL}(\zeta)$ (blue), $\rho_{UR}(1-\zeta)$ (orange), $\rho_{LR}(\zeta)$ (green), and $\rho_{UL}(1-\zeta)$ (red) plotted over $0\leq \zeta\leq 1$ under the SMP mean $\widehat{\psi}$, along with the empirical quantile dependencies (gray). Each panel corresponds to a different variable pair. Asymmetries in a plot indicate asymmetric dependence.
    }
    \label{fig:app_dependence_smi}
\end{figure}

\begin{figure}[tbh]
    \centering
    \includegraphics[width=0.95\linewidth,keepaspectratio]{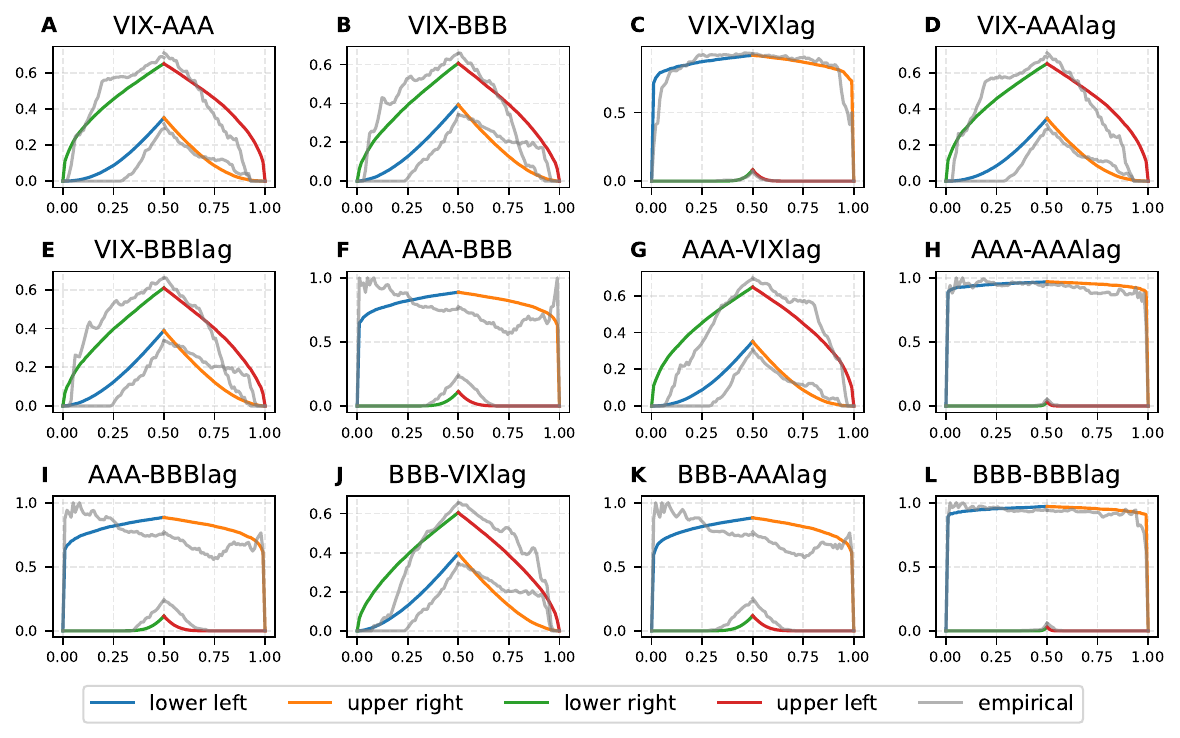}
    \caption{\small Quantile plots from the conventional posterior in the bold yields example. The quantiles are $\rho_{LL}(\zeta)$ (blue), $\rho_{UR}(1-\zeta)$ (orange), $\rho_{LR}(\zeta)$ (green), and $\rho_{UL}(1-\zeta)$ (red) plotted over $0\leq \zeta\leq 1$ under the conventional posterior mean $\widehat{\psi}$, along with the empirical quantile dependencies (gray). Each panel corresponds to a different variable pair. Asymmetries in a plot indicate asymmetric dependence.
    }
    \label{fig:app_dependence_conventional} 
\end{figure}

\begin{figure}[tbh]
    \centering
    \includegraphics[width=0.95\linewidth,keepaspectratio]{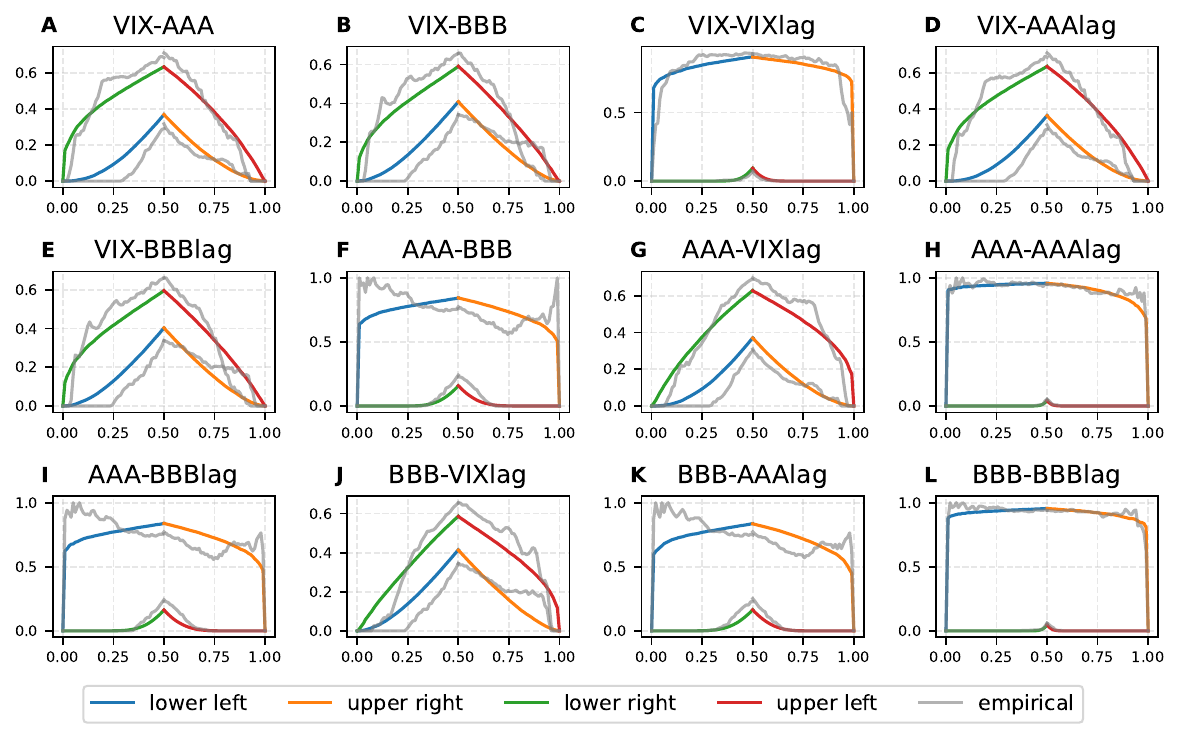}
    \caption{\small  Quantile plots from the fully cut posterior in the bold yields example. The quantiles are $\rho_{LL}(\zeta)$ (blue), $\rho_{UR}(1-\zeta)$ (orange), $\rho_{LR}(\zeta)$ (green), and $\rho_{UL}(1-\zeta)$ (red) plotted over $0\leq \zeta\leq 1$ under the fully cut posterior mean $\widehat{\psi}$, along with the empirical quantile dependencies (gray). Each panel corresponds to a different variable pair. Asymmetries in a plot indicate asymmetric dependence. 
    }
    \label{fig:app_dependence_cut} 
\end{figure}

\end{document}